\documentclass[letter,11pt]{article}
\usepackage{float,graphicx,verbatim,fullpage,hyperref,amssymb,amsmath,amsthm,amsfonts,enumerate,multicol, xspace,comment,xcolor,thm-restate,url, cite}
\usepackage[T1]{fontenc}
\usepackage[margin=1in]{geometry}
\usepackage{algo}
\usepackage{subfig}
\usepackage[margin=1in]{geometry}
\newcommand*{\email}[1]{\texttt{#1}}
\usepackage{hyperref}

\newtheorem{theorem}{Theorem}[section]
\newtheorem{lemma}{Lemma}[section]
\newtheorem{corollary}{Corollary}[section]

\newtheorem{proposition}{Proposition}[section]
\newtheorem{claim}{Claim}[section]

\def\final{0}  
\def\iflong{\iffalse}
\ifnum\final=0  
\newcommand{\cnote}[1]{{\color{blue}[{\tiny Calvin: \bf #1}]\marginpar{*}}}
\newcommand{\wnote}[1]{{\color{cyan}[{\tiny Weihang: \bf #1}]\marginpar{*}}}
\newcommand{\knote}[1]{{\color{red}[{\tiny Karthik: \bf #1}]\marginpar{\color{red}*}}}
\newcommand{\todo}[1]{{\color{red}[{\tiny TODO: \bf #1}]\marginpar{\color{red}*}}}
\else 
\newcommand{\cnote}[1]{}
\newcommand{\wnote}[1]{}
\newcommand{\knote}[1]{}
\newcommand{\todo}[1]{}
\fi  

\newcommand{\R}{\mathbb{R}}

\newcommand{\collection}{\mathcal{C}}
\newcommand{\family}{\mathcal{F}}

\newcommand{\opt}{OPT_{k\text{-cut}}}
\newcommand{\optk}{OPT_{k\text{-cut}}}
\newcommand{\optmm}{OPT_{\text{minmax-}k\text{-partition}}}

\newcommand{\mcp}{\mathcal{P}}
\newcommand{\deltacard}{d}

\newcommand{\minmax}{minmax\xspace}

\newcommand{\enumhkcut}{\textsc{Enum-Hypergraph-$k$-Cut}\xspace}
\newcommand{\enumgkcut}{\textsc{Enum-Graph-$k$-Cut}\xspace}
\newcommand{\enumMMh}{\textsc{Enum-MinMax-Hypergraph-$k$-Partition}\xspace}
\newcommand{\enumMMg}{\textsc{Enum-MinMax-Graph-$k$-Partition}\xspace}
\newcommand{\enumkcutsetreps}{\textsc{Enum-MinMax-Hypergraph-$k$-Cut-Set-Reps}\xspace}

\newcommand{\minkcutset}{\textsc{min-$k$-cut-set}\xspace}
\newcommand{\minmaxkcutset}{\textsc{minmax-$k$-cut-set}\xspace}

\newcommand{\minkcutsets}{\textsc{min-$k$-cut-set}s\xspace}
\newcommand{\minmaxkcutsets}{\textsc{minmax-$k$-cut-set}s\xspace}

\newcommand{\mmh}{\textsc{Minmax-Hypergraph-$k$-Partition}\xspace}
\newcommand{\mmg}{\textsc{Minmax-Graph-$k$-Partition}\xspace}
\newcommand{\mmsubmodkpart}{\textsc{Minmax-Submod-$k$-Partition}\xspace}
\newcommand{\mmsymsubmodkpart}{\textsc{Minmax-SymSubmod-$k$-Partition}\xspace}

\newcommand{\hkcut}{\textsc{Hypergraph-$k$-Cut}\xspace}

\newcommand{\gkcut}{\textsc{Graph-$k$-Cut}\xspace}

\newcommand{\submodkpart}{\textsc{Minsum-Submod-$k$-Partition}\xspace} 

\newcommand{\symsubmodkpart}{\textsc{Minsum-SymSubmod-$k$-Partition}\xspace}

\newcommand{\mypara}[1]{\medskip \noindent {\bf #1}}

\def\complement#1{\overline{#1}}
\def\set#1{\{ #1 \}}
\usepackage[utf8]{inputenc}

\title{Counting and enumerating optimum cut sets for 
hypergraph $k$-partitioning problems for fixed $k$\footnote{University of Illinois, Urbana-Champaign. Email: \email{\{calvinb2, karthe, weihang3\}@illinois.edu}. Supported in part by NSF grants CCF-1814613 and  CCF-1907937.}
}
\author{
Calvin Beideman \and 
Karthekeyan Chandrasekaran \and 
Weihang Wang
}
\date{}

\begin{document}

\maketitle

\begin{abstract}
We consider the problem of enumerating optimal solutions for two hypergraph $k$-partitioning problems---namely, \hkcut and \mmh. The input in hypergraph $k$-partitioning problems is a hypergraph $G=(V, E)$ with positive hyperedge costs along with a fixed positive integer $k$. 
The goal is to find a partition of $V$ into $k$ non-empty parts $(V_1, V_2, \ldots, V_k)$---known as a $k$-partition---so as to minimize an objective of interest. 
\begin{enumerate}
    \item If the objective of interest is the maximum cut value of the parts, then the problem is known as \mmh. A subset of hyperedges is a \minmaxkcutset if it is the subset of hyperedges crossing an optimum $k$-partition for \mmh. 
    \item If the objective of interest is the total cost of hyperedges crossing the $k$-partition, then the problem is known as \hkcut. A subset of hyperedges is a \minkcutset if it is the subset of hyperedges crossing an optimum $k$-partition for \hkcut. 
\end{enumerate}
We give the first polynomial bound on the number of \minmaxkcutsets and a polynomial-time algorithm to enumerate all of them in hypergraphs for every fixed $k$. Our technique is strong enough to also enable an $n^{O(k)}p$-time deterministic algorithm to enumerate all \minkcutsets in hypergraphs, thus improving on the previously known $n^{O(k^2)}p$-time deterministic algorithm, where $n$ is the number of vertices and $p$ is the size of the hypergraph. The correctness analysis of our enumeration approach relies on a structural result that is a strong and unifying generalization of known structural results for \hkcut and \mmh. We believe that our structural result is likely to be of independent interest in the theory of hypergraphs (and graphs).
\end{abstract}

\newpage
\setcounter{page}{1}

\section{Introduction}
\label{sec:intro}


In hypergraph $k$-partitioning problems, the input consists of a hypergraph $G=(V, E)$ with positive hyperedge-costs $c:E\rightarrow \R_+$ and a fixed positive integer $k$ (e.g., $k=2,3,4, \ldots$). The goal is to find a partition of the vertex set into $k$ \emph{non-empty} parts $V_1, V_2, \ldots, V_k$ so as to minimize an objective of interest. 
There are several natural objectives of interest in hypergraph $k$-partitioning problems. In this work, we focus on two particular objectives: \mmh and \hkcut: 
\begin{enumerate}
\item In \mmh, the objective is to minimize the maximum cut value of the parts of the $k$-partition---i.e., minimize $\max_{i=1}^k c(\delta(V_i))$; here $\delta(V_i)$ is the set of hyperedges intersecting both $V_i$ and $V\setminus V_i$ 
and $c(\delta(V_i))=\sum_{e\in \delta(V_i)}c(e)$ is the total cost of hyperedges in $\delta(V_i)$. 
\item In \hkcut\footnote{We emphasize that the objective of \hkcut is not equivalent to minimizing $\sum_{i=1}^k c(\delta(V_i))$. 
}, the objective is to minimize the cost of hyperedges crossing 
the $k$-partition---i.e., minimize $c(\delta(V_1, \ldots, V_k))$; here $\delta(V_1, \ldots, V_k)$ is the set of hyperedges that intersect at least two sets in $\{V_1, \ldots, V_k\}$ and $c(\delta(V_1, \ldots, V_k))=\sum_{e\in \delta(V_1, \ldots, V_k)}c(e)$ is the total cost of hyperedges in $\delta(V_1, \ldots, V_k)$. 
\end{enumerate}
If the input $G$ is a graph, then we will refer to these problems as \mmg and \gkcut respectively. 
We note that the case of $k=2$ corresponds to global minimum cut in both objectives. 
In this work, we focus on the problem of enumerating all optimum solutions to \mmh and \hkcut.

\mypara{Motivations and Related Problems.} 
We consider the problem of counting and enumerating optimum solutions for partitioning problems over hypergraphs for three reasons. 
Firstly, 
hyperedges 
provide more powerful modeling capabilities than edges and consequently, several problems in hypergraphs become non-trivial in comparison to graphs. 
Although hypergraphs and partitioning problems over hypergraphs (including \mmh) were discussed as early as 1973 by Lawler \cite{La73}, most of these problems still remain open. 
The powerful modeling capability of hyperedges has been useful in a variety of modern applications, which in turn, has led to a resurgence in the study of hypergraphs with recent works focusing on min-cuts, cut-sparsifiers, spectral-sparsifiers, etc. 
\cite{KK15, ChekuriX18, GKP17, CXY19, FPZ19, CKN20, KKTY21, SY19, BTS19, CC20, BCW22}. Our work adds to this rich and emerging theory of hypergraphs. 

Secondly, hypergraph $k$-partitioning problems are special cases of submodular $k$-partitioning problems. In submodular $k$-partitioning problems, the input is a finite ground set $V$, a submodular function\footnote{A real-valued set function $f:2^V\rightarrow \R$ is submodular if $f(A) + f(B) \ge f(A\cap B) + f(A\cup B)$ 
$\forall$ 
$A, B\subseteq V$.} $f:2^V\rightarrow \R$ provided by an evaluation oracle\footnote{An evaluation oracle for a set function $f$ over a ground set $V$ returns the value of $f(S)$ given $S \subseteq V$.} and a positive integer $k$ (e.g., $k=2,3,4,\ldots$). The goal is to partition the ground set $V$ into $k$ non-empty parts $V_1, V_2, \ldots, V_k$ so as to minimize an objective of interest. Two natural objectives are of interest: (1) In \mmsubmodkpart, the objective is to minimize $\max_{i=1}^k f(V_i)$ and (2) In \submodkpart, the objective is to minimize $\sum_{i=1}^k f(V_i)$. If the given submodular function is symmetric\footnote{A real-valued set function $f:2^V\rightarrow \R$ is symmetric if $f(A) = f(V\setminus A)$ 
$\forall$ 
$A\subseteq V$.}, then we denote the resulting problems as \mmsymsubmodkpart and \symsubmodkpart respectively. Since the hypergraph cut function is symmetric submodular, it follows that \mmh is a special case of \mmsymsubmodkpart. Moreover, \hkcut is a special case of \submodkpart (this reduction is slightly non-trivial with the submodular function in the reduction being asymmetric---e.g., see \cite{OFN12} for the reduction). Queyranne claimed, in 1999, a polynomial-time algorithm for \symsubmodkpart for every fixed $k$ \cite{Q99}, however the claim was retracted subsequently (see \cite{GQ}). The complexity status of submodular $k$-partitioning problems (for fixed $k\ge 4$) are open, so recent works have focused on hypergraph $k$-partitioning problems as a stepping stone towards submodular $k$-partitioning \cite{ZNI05, Zhthesis, OFN12, GQ, CC20, CC21, BCW22}. Our work contributes to this stepping stone by advancing the state of the art in hypergraph $k$-partitioning problems. 
We emphasize that the complexity status of two other variants of hypergraph $k$-partitioning problems which are also special cases of \submodkpart are still open (see \cite{ZNI05, Zhthesis, OFN12} for these variants). 


Thirdly, counting and enumeration of optimum solutions for \emph{graph} $k$-partitioning problems are fundamental to graph theory and extremal combinatorics. 
They have found farther reaching applications than initially envisioned. We discuss some of the results and applications for $k=2$ and $k>2$ now. 
For $k=2$ in connected graphs, 
it is well-known that the number of min-cuts and the number of $\alpha$-approximate min-cuts are at most $\binom{n}{2}$ and $O(n^{2\alpha})$ respectively, and they can all be enumerated in polynomial time for constant $\alpha$. 
These combinatorial results have been the crucial ingredients of several algorithmic and representation results in graphs. 
On the algorithmic front, these results enable fast randomized construction of graph skeletons which, in turn, plays a crucial role in fast algorithms to solve graph min-cut \cite{Karger00}. 
On the representation front, counting results form the backbone of cut sparsifiers which in turn have found applications in sketching and streaming \cite{AG09, AGM12, AGM12-pods, KK15}. 
A polygon representation of the family of $6/5$-approximate min-cuts in graphs was given by Benczur and Goemans in 1997 (see \cite{Ben95, Benthesis, BG08})---this representation was used in the recent groundbreaking 
$(3/2-\epsilon)$-approximation for metric TSP \cite{KKG21}. 
On the approximation front, in addition to the $(3/2-\epsilon)$-approximation for metric TSP \cite{KKG21}, counting results also led to the recent $1.5$-approximation for path TSP \cite{Zen19}. 
For $k>2$, we note that fast algorithms for \gkcut have been of interest since they help in generating cutting planes while solving TSP \cite{ccps-book, abcc06-book}. 
A recent series of works aimed towards improving the bounds on the number of optimum solutions for \gkcut 
culminated in a drastic improvement in the run-time to solve \gkcut \cite{GLL19-STOC, GLL20-STOC, GHLL20}.
Given the status of counting and enumeration results for $k$-partitioning in graphs and their algorithmic and representation implications that were discovered subsequently, we believe that a similar understanding 
in hypergraphs could serve as an important ingredient in the 
algorithmic and representation theory of hypergraphs. 

\mypara{The Enumeration Problem.} 
There is a fundamental structural distinction between hypergraphs and graphs that becomes apparent while enumerating optimum solutions to $k$-partitioning problems. 
In connected graphs, the number of optimum $k$-partitions for \gkcut and for \mmg are $n^{O(k)}$ and $n^{O(k^2)}$ respectively and they can all be enumerated in polynomial time, where $n$ is the number of vertices in the input graph \cite{KS96, Th08, CQX20, GLL20-STOC, GHLL20, CW21}. 
In contrast, a connected hypergraph could have exponentially many optimum $k$-partitions for both \mmh and \hkcut even for $k=2$---e.g., consider the hypergraph with a single hyperedge containing all vertices; we will denote this as the spanning-hyperedge-example. Hence, enumerating all optimum $k$-partitions for hypergraph $k$-partitioning problems in polynomial time is impossible. 
Instead, our goal in the enumeration problems is to enumerate \emph{$k$-cut-sets} corresponding to optimum $k$-partitions. We will call a subset $F\subseteq E$ of hyperedges to be a $k$-cut-set if there exists a $k$-partition $(V_1, \ldots, V_k)$ such that $F=\delta(V_1, \ldots, V_k)$; we will call a $2$-cut-set as a cut-set. 
In the enumeration problems that we will consider, the input consists of a hypergraph $G=(V, E)$ with positive hyperedge-costs $c:E\rightarrow \R_+$ and a fixed positive integer $k$ (e.g., $k=2,3,4,\ldots$).
\begin{enumerate}
    \item For an optimum $k$-partition $(V_1, \ldots, V_k)$ for \mmh in $(G, c)$, we will denote $\delta(V_1, \ldots, V_k)$ as a \minmaxkcutset. In \enumMMh, the goal is to enumerate all \minmaxkcutsets.
    \item For an optimum $k$-partition $(V_1, \ldots, V_k)$ for \hkcut in $(G, c)$, we will denote $\delta(V_1, \ldots, V_k)$ as a \minkcutset. In \enumhkcut, the goal is to enumerate all \minkcutsets. 
\end{enumerate}
We observe that in the spanning-hyperedge-example, although the number of optimum $k$-partitions for \mmh (as well as  \hkcut) is exponential, the number of \minmaxkcutset{s} (as well as \minkcutset{s}) is only one. 

\subsection{Results}
In contrast to graphs, whose representation size is the number of edges, the representation size of a hypergraph $G=(V, E)$ is $p:=\sum_{e\in E}|e|$. Throughout, our algorithmic discussion will focus on the case of fixed $k$ (e.g., $k=2,3,4,\ldots$). 

There are no prior results regarding \enumMMh in the literature. We recall the status of \mmh. 
As mentioned earlier, \mmh was discussed as early as 1973 by Lawler \cite{La73} with its complexity status being open until recently. 
We note that the objective here could be viewed as aiming to find a \emph{fair} $k$-partition, i.e., a $k$-partition where no part pays too much in cut value. Motivated by this connection to fairness, Chandrasekaran and Chekuri (2021) \cite{CC21} studied the more general problem of \mmsymsubmodkpart. They gave the first (deterministic) polynomial-time algorithm to solve \mmsymsubmodkpart and as a consequence, obtained the first polynomial-time algorithm to solve \mmh. Their algorithm does not show any bound on the number of \minmaxkcutset{s} since it solves the more general problem of \mmsymsubmodkpart for which the number of optimum $k$-partitions can indeed be exponential (recall the spanning-hyperedge-example). 
Focusing on hypergraphs raises the question of whether all $k$-cut-sets corresponding to optimum solutions can be enumerated efficiently for every fixed $k$. We answer this question affirmatively by giving 
the first polynomial-time algorithm for \enumMMh. 

\begin{theorem}\label{thm:minmax-enumeration}
There exists a deterministic algorithm to solve \enumMMh that runs in time $O(kn^{4k^2-2k+1}p)$, where $n$ is the number of vertices and $p$ is the size of the input hypergraph. Moreover, the number of \minmaxkcutset{s} in a $n$-vertex hypergraph is $O(n^{4k^2-2k})$. 
\end{theorem}

We emphasize that 
our result shows the first polynomial bound on the number of \minmaxkcutsets in hypergraphs for every fixed $k$ (in addition to a polynomial-time algorithm to enumerate all of them for every fixed $k$). 
Our upper bound of $n^{O(k^2)}$ on the number of \minmaxkcutset{s} is tight---there exist $n$-vertex connected \emph{graphs} for which the number of \minmaxkcutset{s} is $n^{\Theta(k^2)}$ (see Section \ref{sec:lower-bound}). 

\medskip
Next, we briefly recall the status of \hkcut and \enumhkcut. 
\hkcut was shown to be solvable in randomized polynomial time only recently \cite{CXY19, FPZ19}; the randomized algorithms also showed that the number of \minkcutset{s} is $O(n^{2k-2})$ and they can all be enumerated in randomized polynomial time. A subsequent deterministic algorithm was designed to solve \hkcut in time $n^{O(k)}p$ by Chandrasekaran and Chekuri \cite{CC20}. Chandrasekaran and Chekuri's techniques were extended to design the first deterministic polynomial-time algorithm to solve \enumhkcut in \cite{BCW22}. The algorithm for \enumhkcut given in \cite{BCW22} runs  in time  $n^{O(k^2)}p$. 
We note that this run-time has a 
quadratic dependence on $k$ in the exponent of $n$ although the number of \minkcutset{s} has only linear dependence on $k$ in the exponent of $n$ (since it is $O(n^{2k-2})$). So, an open question that remained from \cite{BCW22} is whether one can obtain an $n^{O(k)}p$-time deterministic algorithm for \enumhkcut. We resolve this question affirmatively.

\begin{theorem} \label{thm:k-cut-set-enumeration}
There exists a deterministic algorithm to solve \enumhkcut that runs in time $O(n^{16k-25}p)$, where $n$ is the number of vertices and $p$ is the size of the input hypergraph. 
\end{theorem}

Our algorithms for both \enumMMh and \enumhkcut are based on a structural theorem that allows for efficient recovery of optimum $k$-cut-sets via minimum $(s,t)$-terminal cuts (see Theorem \ref{thm:cut-set-recovery}). 
Our structural theorem builds on structural theorems that have appeared in previous works on \mmh and \hkcut \cite{CC20, CC21, BCW22}. 
Our structural theorem may appear to be natural/incremental in comparison to ones that appeared in previous works, but formalizing the theorem and proving it is a significant part of our contribution. Moreover, 
our single structural theorem is strong enough to enable efficient algorithms for both \enumhkcut as well as \enumMMh  in contrast to previously known structural theorems. 
In this sense, our structural theorem can be viewed as a strong and unifying generalization of structural theorems that have appeared in previous works. 
We believe that our structural theorem 
will be of independent interest in the theory of cuts and partitioning in hypergraphs (as well as graphs). 

\subsection{Technical overview and main structural result}
\label{sec:techniques}
We focus on the unit-cost variant of \enumhkcut and \enumMMh in the rest of this work for the sake of notational simplicity---i.e., the cost of every hyperedge is $1$. Throughout, we will allow
multigraphs and hence, this is without loss of generality. Our
algorithms extend in a straightforward manner to arbitrary hyperedge
costs. They rely only on minimum $(s,t)$-terminal cut computations and hence, they are strongly polynomial-time algorithms.


\mypara{Notation and background.} 
Let $G=(V,E)$ be a hypergraph. Throughout this work, 
$n$ will denote the number of vertices in $G$, $m$ will denote the number of hyperedges in $G$, and $p:=\sum_{e\in E}|e|$ will denote the representation size of $G$. 
We will denote a partition of the vertex set into $h$ non-empty parts by an ordered tuple $(V_1, \ldots, V_h)$ and call such an ordered tuple as an $h$-partition. 
For a partition
$\mcp=(V_1,V_2,\ldots,V_h)$, 
we will say that a hyperedge $e$ crosses the partition $\mcp$ if it intersects at least two parts of the partition. 
We will refer to a $2$-partition as a cut. 
For a non-empty proper subset $U$ of vertices,
we will use $\complement{U}$ to denote $V\setminus U$, 
$\delta(U)$ to denote the set of hyperedges crossing the cut $(U, \complement{U})$, 
and $\deltacard(U):=|\delta(U)|$ to denote the cut value of $U$. 
We observe that $\delta(U)=\delta(\complement{U})$, so we will use $\deltacard(U)$ to denote the value of the cut $(U, \complement{U})$. 
More generally, given a partition
$\mcp=(V_1,V_2,\ldots,V_h)$, 
we denote 
the 
set 
of hyperedges crossing the partition by 
$\delta(V_1, V_2, \ldots, V_h)$ 
(also by $\delta(\mcp)$ for brevity)
and the number of hyperedges crossing the partition by 
$|\delta(V_1, V_2, \ldots, V_h)|$.  
We will denote the optimum value of \mmh and \hkcut respectively by 
\begin{align*}
\optmm&:=\min\left\{\max_{i\in [k]}|\delta(V_i)|: (V_1, \ldots, V_k) \text{ is a $k$-partition of $V$}\right\}  \text{ and} \\
\optk&:=\min\left\{\left|\delta(V_1, \ldots, V_k)\right|: (V_1, \ldots, V_k) \text{ is a $k$-partition of $V$}\right\}. 
\end{align*}
A key algorithmic tool will be the use of fixed-terminal cuts.  Let $S$, $T$ be
disjoint non-empty subsets of vertices. A $2$-partition $(U, \complement{U})$ is
an $(S,T)$-terminal cut if $S\subseteq U\subseteq V\setminus T$. Here,
the set $U$ is known as the source set and the set $\complement{U}$ is
known as the sink set.  A minimum-valued $(S,T)$-terminal cut is known
as a \emph{minimum $(S,T)$-terminal cut}.  Since there could be multiple
minimum $(S,T)$-terminal cuts, we will be interested in \emph{source
  minimal} minimum $(S,T)$-terminal cuts. 
For every pair of disjoint non-empty subsets $S$ and $T$ of vertices, there exists a unique source minimal minimum $(S,T)$-terminal cut and it can be found
in deterministic polynomial time via standard maxflow algorithms.
In particular, the source minimal minimum $(S, T)$-terminal cut can be found in time $O(np)$ \cite{ChekuriX18}. 



Our technique to enumerate all \minmaxkcutsets and all \minkcutsets will build on the approaches of  Chandrasekaran and Chekuri for \hkcut and \mmsymsubmodkpart \cite{CC20, CC21, BCW22}. 
We need the following structural theorem that was shown in \cite{BCW22}. 
\begin{restatable}{theorem}{thmStructureOne}
\label{theorem: structure thm 1}\cite{BCW22}
Let $G=(V,E)$ be a hypergraph and let $\opt$ be the optimum value of \hkcut in $G$ for some integer $k\ge 2$. Suppose $(U,\complement{U})$ is a $2$-partition of $V$ with $d(U)<\opt$. Then, for every pair of vertices $s\in U$ and $t\in\overline{U}$, there exist subsets  $S\subseteq U\setminus \{s\}$ and $T\subseteq \complement{U}\setminus \{t\}$ with $|S|\le 2k-3$ and $|T|\leq 2k-3$ such that $(U,\complement{U})$ is the unique minimum $(S\cup\{s\},T\cup\{t\})$-terminal cut in $G$.
\end{restatable}

\mypara{\enumhkcut.} We first focus on \enumhkcut. We note that Theorem \ref{theorem: structure thm 1} will allow us to recover those parts $V_i$ of an optimum $k$-partition $(V_1, \ldots, V_k)$ for which $d(V_i)<\optk$. However, recall that our goal is \emph{not} to recover all optimum $k$-partitions for \hkcut, but rather to recover all \minkcutset{s} (i.e., not to recover the parts of every optimum $k$-partition, but rather only to recover the $k$-cut-set of every optimum $k$-partition). The previous work \cite{BCW22} that designed an $n^{O(k^2)}p$-time deterministic enumeration algorithm achieved this by proving the following structural result: suppose $(V_1, \ldots, V_k)$ is an optimum $k$-partition for \hkcut for which $d(V_1)=\optk$. Then, they showed that for every subset $T\subseteq \complement{V_1}$ satisfying $T\cap V_j\neq \emptyset$ for all $j\in \{2, \ldots, k\}$, there exists a subset $S\subseteq V_1$ with $|S|\le 2k$ such that the source minimal minimum $(S, T)$-terminal cut $(A, \complement{A})$ satisfies $\delta(A) = \delta(V_1)$. This structural theorem in conjunction with Theorem \ref{theorem: structure thm 1} allows one to enumerate a candidate family $\mathcal{F}$ of $n^{O(k^2)}$ subsets of hyperedges such that every \minkcutset is present in the family. The drawback of their structural theorem is that it is driven towards recovering the cut-set $\delta(V_i)$ of every part $V_i$ of every optimum $k$-partition $(V_1, \ldots, V_k)$. Hence, their algorithmic approach ends up with a run-time of $n^{O(k^2)}p$. 
In order to improve the run-time, we prove a stronger result: we show that for an 
\emph{arbitrary} cut $(U, \complement{U})$ with cut value $\optk$ (as opposed to only those sets $V_i$ of an optimum $k$-partition $(V_1, \ldots, V_k)$), its cut-set $\delta(U)$ can be recovered as the cut-set of \emph{any} minimum $(S, T)$-terminal cut for some $S$ and $T$ of small size. 
The following is the main structural theorem of this work. 


\begin{restatable}{theorem}{thmStrongerStructure}
\label{thm:cut-set-recovery}
Let $G=(V,E)$ be a hypergraph and let $\opt$ be the optimum value of \hkcut in $G$ for some integer $k\ge 2$. Suppose $(U,\complement{U})$ is a $2$-partition of $V$ with $d(U) = \opt$.
Then, there exist sets $S \subseteq U$, $T \subseteq \overline{U}$ with $|S| \leq 2k-1$ and $|T| \leq 2k-1$ such that every minimum $(S,T)$-terminal cut $(A, \overline{A})$ satisfies $\delta(A) = \delta(U)$.
\end{restatable}

We encourage the reader to compare and contrast 
Theorems \ref{theorem: structure thm 1} and \ref{thm:cut-set-recovery}. The former helps to recover cuts whose cut value is strictly smaller than $\optk$ while the latter helps to recover \emph{cut-sets} whose size is equal to $\optk$. So, the latter theorem is weaker since it only recovers cut-sets, but we emphasize that this is the best possible that one can hope to do (as seen from the spanning-hyperedge-example). However, proving the latter theorem requires us to work with cut-sets (as opposed to cuts) which is a technical barrier to overcome. Indeed, our proof of Theorem \ref{thm:cut-set-recovery} deviates significantly from the proof of Theorem \ref{theorem: structure thm 1} since we have to work with cut-sets. Our proof also deviates from the structural result in \cite{BCW22} that was mentioned in the paragraph above Theorem \ref{thm:cut-set-recovery} since our result is stronger than their result---our result helps to recover the cut-set $\delta(U)$ of an arbitrary cut $(U, \complement{U})$ whose cut value is $d(U)=\optk$ while their result helps only to recover the cut-set $\delta(V_i)$ of a part $V_i$ of an optimum $k$-partition $(V_1, \ldots, V_k)$ for \hkcut whose cut value is $d(V_i)=\optk$; moreover, their proof technique crucially relies on a containment property with respect to the part $V_i$, 
whereas under the hypothesis of our structural theorem, the containment property fails with respect to the set $U$ and consequently, our proof technique differs from theirs.


Theorems \ref{theorem: structure thm 1} and \ref{thm:cut-set-recovery} lead to a deterministic $n^{O(k)}$-time algorithm to enumerate all \minkcutset{s} via a divide-and-conquer approach. We describe this algorithm now: 
For each pair $(S, T)$ of disjoint subsets of vertices $S$ and $T$ with $|S|$, $|T|\le 2k-1$, compute the source minimal minimum $(S, T)$-terminal cut $(A, \complement{A})$; (i) if $G-\delta(A)$ has at least $k$ connected components, then add $\delta(A)$ to the candidate family $\family$; (ii) otherwise, add the set $A$ to a collection $\collection$. We note that the sizes of the family $\family$ and the collection $\collection$ are $O(n^{4k-2})$. Next, for each subset $A$ in the collection $\collection$, recursively enumerate all \textsc{min-$k/2$-cut-set}s 
in the subhypergraphs induced by $A$ and $\complement{A}$ respectively\footnote{Subhypergraph $G[A]$ has vertex set $A$ and contains all hyperedges of $G$ which are entirely contained within $A$.}---denoted $G[A]$ and $G[\complement{A}]$ respectively---and add $\delta(A)\cup F_1\cup F_2$ to the family $\family$ for each $F_1$ and $F_2$ being \textsc{min-$k/2$-cut-set} in $G[A]$ and $G[\complement{A}]$ respectively. Finally, return the subfamily of $k$-cut-sets from the family $\family$ that are of smallest size. 

We sketch the correctness analysis of the above approach: let $F=\delta(V_1, \ldots, V_k)$ be a \minkcutset with $(V_1, \ldots, V_k)$ being an optimum $k$-partition for \hkcut. We will show that the family $\family$ contains $F$. Let $U:=\cup_{i=1}^{k/2}V_i$. We note that $\delta(U)\subseteq F$. 
We have two possibilities: (1) Say $d(U)=|F|$. Then, $d(U)=\optk$. Consequently, by Theorem \ref{thm:cut-set-recovery}, 
the \minkcutset $F$ will be added to the family $\family$ by step (i). (2) Say $d(U)<|F|$. Then, by Theorem \ref{theorem: structure thm 1}, the set $U=\cup_{i=1}^{k/2}V_i$ will be added to the collection $\collection$ by step (ii); moreover, $F_1:=F\cap E(G[U])$ and $F_2:=F\cap E(G[\complement{U}])$ are \textsc{min-$k/2$-cut-sets} in $G[U]$ and $G[\complement{U}]$ respectively and they would have been enumerated by recursion, and hence, the set $\delta(U)\cup F_1 \cup F_2=F$ will be added to the family $\family$. The size of the family $\family$ can be  shown to be $n^{O(k\log{k})}$ and the run-time  is $n^{O(k\log{k})}p$ (see Theorem \ref{theorem: DC enum algorithm}). Using the known fact that the number of \minkcutset{s} in a $n$-vertex hypergraph is $O(n^{2k-2})$, we can 
improve the run-time analysis of this approach to $n^{O(k)}p$ (see Lemma \ref{lemma:DC-better-run-time}). 

\mypara{\enumMMh.} Next, we focus on \enumMMh. There is a fundamental technical issue in enumerating \minmaxkcutset{s} as opposed to \minkcutset{s}. We highlight this technical issue now. Suppose we find an optimum $k$-partition $(V_1, \ldots, V_k)$ for \mmh (say via Chandrasekaran and Chekuri's algorithm \cite{CC21}) and store only the \minmaxkcutset $F=\delta(V_1, \ldots, V_k)$ but forget to store the partition $(V_1, \ldots, V_k)$; now, by knowing a \minmaxkcutset $F$, can we recover \emph{some} optimum $k$-partition for \mmh (not necessarily $(V_1, \ldots, V_k)$)? Or by knowing a \minmaxkcutset $F$, is it even possible to find 
the value $\optmm$ 
without solving \mmh from scratch again---i.e., is there an advantage to knowing a \minmaxkcutset in order to solve \mmh? We are not aware of such an advantage. This is in stark contrast to \hkcut where knowing a \minkcutset enables a linear-time solution to \hkcut\footnote{Suppose we know a \minkcutset $F$. Then consider the connected components $Q_1, \ldots, Q_t$ in $G-F$ and create a partition $(P_1, \ldots, P_k)$ by taking $P_i=Q_i$ for every $i\in [k-1]$ and $P_k=\cup_{j=k}^t Q_j$; such a $k$-partition $(P_1, \ldots, P_k)$ will be an optimum $k$-partition for \hkcut.}. 

Why is this issue significant while solving \enumMMh? 
We recall that in our approach for \enumhkcut, the algorithm computed a polynomial-sized family $\family$ containing all \minkcutsets and returned the ones with smallest size---the smallest size ones will exactly be \minkcutsets. It is unclear if a similar approach could work for enumerating \minmaxkcutsets: suppose we do have an algorithm to enumerate a polynomial-sized family $\family$ containing all \minmaxkcutset{s}; now, in order to return all \minmaxkcutset{s} (which is a subfamily of $\family$), note that we need to identify them among the ones in the family $\family$---i.e., 
we need to verify if a given subset $F\in \family$ of hyperedges is a \minmaxkcutset; this verification problem is closely related to the question mentioned in the previous paragraph. We do not know how to address this verification problem directly. So, our algorithmic approach for \enumMMh has to overcome this technical issue. 

Our ingredient to overcome this technical issue is to enumerate \emph{representatives} for \minmaxkcutsets. 
For a $k$-partition $(V_1, \ldots, V_k)$ and disjoint subsets $U_1, \ldots, U_k\subseteq V$, we will call the $k$-tuple $(U_1, \ldots, U_k)$ to be a \emph{$k$-cut-set representative} of $(V_1, \ldots, V_k)$ if $U_i\subseteq V_i$ and $\delta(U_i)=\delta(V_i)$ for all $i\in [k]$. 
We note that a fixed $k$-partition $(V_1, \ldots, V_k)$ could have several $k$-cut-set representatives and a fixed $k$-tuple $(U_1, \ldots, U_k)$ could be the $k$-cut-set representative of several $k$-partitions. 
Yet, it is possible to efficiently verify 
if a given $k$-tuple $(U_1, \ldots, U_k)$ is a $k$-cut-set representative (Theorem \ref{Thm: hyp minmax-1}). 
Moreover, knowing a $k$-cut-set representative $(U_1, \ldots, U_k)$ of a $k$-partition $(V_1, V_2, \ldots, V_k)$ allows one to recover the $k$-cut-set $F:=\delta(V_1, \ldots, V_k)$ since $F=\cup_{i=1}^k \delta(U_i)$. 
Thus, in order to enumerate all \minmaxkcutsets, it suffices to enumerate $k$-cut-set representatives of \emph{all} optimum $k$-partitions for \mmh. 
At this point, the astute reader may wonder 
if there exists a polynomial-sized family of $k$-cut-set representatives of all optimum $k$-partitions for \mmh given that the number of optimum $k$-partitions for \mmh could be exponential. 
For example, is there a polynomial-sized family of $k$-cut-set representatives of all optimum $k$-partitions for \mmh in the spanning-hyperedge-example? 
Indeed, in the spanning-hyperedge-example, even though the number of optimum $k$-partitions for \mmh is exponential, there exists a $(k!\binom{n}{k})$-sized family of $k$-cut-set representatives of all optimum $k$-partitions: 
consider the family $\{(\{v_1\}, \ldots, \{v_k\}): v_1, \ldots, v_k\in V, v_i\neq v_j\ \forall\ \text{distinct } i, j\in [k]\}$. 

It turns out that Theorems \ref{theorem: structure thm 1} and \ref{thm:cut-set-recovery} are strong enough to enable efficient enumeration of 
$k$-cut-set representatives of all optimum $k$-partitions for \mmh. We describe the algorithm to achieve this: 
For each pair $(S, T)$ of disjoint subsets of vertices with $|S|$, $|T|\le 2k-1$, compute the source minimal minimum $(S, T)$-terminal cut $(U, \complement{U})$ and add $U$ to a candidate collection $\collection$. We note that the size of the collection $\collection$ is $O(n^{4k-2})$. Next, for each $k$-tuple $(U_1, \ldots, U_k)\in \collection^k$, verify if $(U_1, \ldots, U_k)$ is a $k$-cut-set representative (using Theorem \ref{Thm: hyp minmax-1}) and if so, then add the $k$-tuple to the candidate family $\mathcal{D}$. Finally, return $\arg\min\{\max_{i=1}^k d(U_i): (U_1, \ldots, U_k)\in \mathcal{D}\}$, i.e., prune and return the subfamily of $k$-cut-set representatives $(U_1, \ldots, U_k)$ from the family $\mathcal{D}$ that have minimum $\max_{i=1}^k d(U_i)$. 

We note that the size of the family $\mathcal{D}$ is $n^{O(k^2)}$ and consequently, the run-time is $n^{O(k^2)}p$. We sketch the correctness analysis of the above approach: let $(V_1, \ldots, V_k)$ be an optimum $k$-partition for \mmh. We will show that the family $\mathcal{D}$ contains a $k$-cut-set representative of $(V_1, \ldots, V_k)$. By noting that $\optmm\le \optk$ and by Theorems \ref{theorem: structure thm 1} and \ref{thm:cut-set-recovery}, for every $i\in [k]$, we have a set $U_i$ in the collection $\collection$ with $U_i\subseteq V_i$ and $\delta(U_i)=\delta(V_i)$. Hence, the $k$-tuple $(U_1, \ldots, U_k)\in \collection^k$ is a $k$-cut-set representative and it will be added to the family $\mathcal{D}$. The final pruning step will not remove $(U_1, \ldots, U_k)$ from the family $\mathcal{D}$ and hence, it will be in the subfamily returned by the algorithm. 

\mypara{Significance of our technique.}
As mentioned earlier, our techniques build on the structural theorems that appeared in previous works \cite{CC20, CC21, BCW22}. The main technical novelty of our contribution lies in Theorem \ref{thm:cut-set-recovery} which can be viewed as the culmination of structural theorems developed in those previous works. 
We also emphasize that using minimum $(s,t)$-terminal cuts to solve global partitioning problems is not a new technique per se (e.g., minimum $(s, t)$-terminal cut is the first and most natural approach to solve global minimum cut). This technique of using minimum $(s,t)$-terminal cuts to solve global partitioning problems has a rich variety of applications 
in combinatorial optimization: 
e.g., 
(1) it was used to design the first efficient algorithm for \gkcut for fixed $k$ \cite{GH94}, 
(2) it was used to design efficient algorithms for certain constrained submodular minimization  problems \cite{GR95, NSZ19},
and (3) more recently, it was used to design fast algorithms for global minimum cut in graphs as well as to obtain fast Gomory-Hu trees in unweighted graphs \cite{LP20, AKT21}. 
The applicability of this technique relies on identifying and proving appropriate structural results. Our Theorem \ref{thm:cut-set-recovery} is such a structural result. 
The merit of the structural result lies in its ability to solve two different enumeration problems in hypergraph $k$-partitioning which was not possible via structural theorems that were developed before.  
Moreover, it leads to the first polynomial bound on the number of \minmaxkcutsets in hypergraphs for every fixed $k$. 



\mypara{Organization.} 
We discuss related work in Section \ref{sec:related-work}. 
In Section \ref{sec:prelims}, we recall properties of the hypergraph cut function. 
In Section \ref{sec:cut-set-recovery-part-1}, we prove a special case of Theorem \ref{thm:cut-set-recovery}. In Section \ref{sec:cut-set-recovery}, we use this special case to prove  Theorem \ref{thm:cut-set-recovery}. 
In Section \ref{sec:enumhkcut-algo}, we design an efficient algorithm for \enumhkcut and prove Theorem \ref{thm:k-cut-set-enumeration}. In Section \ref{sec:enummmh-algo}, we design an efficient algorithm for \enumMMh and prove Theorem \ref{thm:minmax-enumeration}. We present a lower bound example in Section \ref{sec:lower-bound}. We conclude with an open question in Section \ref{sec:conclusion}. 

\subsection{Related work}
\label{sec:related-work}
We briefly discuss the status of the enumeration problems for $k=2$ followed by the status for $k\ge 2$ in graphs and hypergraphs. 


\mypara{Enumeration problems for $k=2$.} For $k=2$, both \enumMMh and \enumhkcut are equivalent to enumerating optimum solutions for global minimum cut in hypergraphs. 
For graphs that are connected, it is well-known that the number of minimum cuts and the number of $\alpha$-approximate minimum cuts are at most $\binom{n}{2}$ and $O(n^{2\alpha})$ respectively and they can all be enumerated in polynomial time for constant $\alpha$ \cite{DKL76, Kar93, GR95, NNI97, HW96, CQX20}. For connected hypergraphs, the number of minimum cuts can be exponential as seen from the spanning-hyperedge-example. However, recent results have shown that the number of minimum cut-sets in a hypergraph is at most $\binom{n}{2}$ via several different techniques and they can all be enumerated in polynomial time \cite{ChekuriX18, GKP17, CXY19, FPZ19, BCW22}. On the other hand, there exist hypergraphs with exponential number of $2$-approximate minimum cut-sets.\footnote{Consider the $n$-vertex hypergraph $G=(V, E)$ obtained from the complete bipartite graph $K_{n/2, n/2}=(L\cup R, E')$ by adding $n^2/4$ copies of two hyperedges $e_1$ and $e_2$, where $e_1:=L$  and $e_2:=R$.}

\mypara{\gkcut and \enumgkcut.} 
When $k$ is part of input, \gkcut is NP-hard \cite{GH94} and admits a $2(1-1/k)$-approximation \cite{SV95, RS08, ZNI05, GLL18-FOCS, Q19}. Manurangsi \cite{Ma18} showed
that there is no polynomial-time $(2-\epsilon)$-approximation for any
constant $\epsilon>0$ assuming the \emph{Small Set Expansion
  Hypothesis} \cite{RS10}. 
  We note that \gkcut is $W[1]$-hard when parameterized by $k$ \cite{DEFPR03} and admits a fixed-parameter approximation scheme when parameterized by $k$ \cite{GLL18-SODA, GLL18-FOCS, KL20, LSS20}, and 
  is fixed-parameter tractable when parameterized by $k$ and the solution size \cite{KT11}.

\gkcut for fixed $k$ was shown to be polynomial-time solvable by Goldschmidt and Hochbaum \cite{GH94}. Subsequently, Karger and Stein \cite{KS96} gave a randomized polynomial-time algorithm whose analysis also showed that the number of optimum $k$-partitions in a connected graph is $O(n^{2k-2})$ and they can all be enumerated in polynomial time for every fixed $k$ (also see \cite{KYN07, Xi08, Th08, CQX20}). The number of optimum $k$-partitions has recently been improved to $O(n^k)$ for fixed $k$ thereby leading to a faster algorithm for \gkcut for fixed $k$ \cite{GLL19-STOC, GLL20-STOC, GHLL20}.

\mypara{\hkcut and \enumhkcut.} 
When $k$ is part of input, \hkcut is at least as hard as the \emph{densest $k$-subgraph} problem \cite{CL20}. Combined with results in \cite{Ma17dks}, this implies that \hkcut is unlikely to have a sub-polynomial factor approximation ratio. 
Moreover, \hkcut is $W[1]$-hard even when parameterized by $k$
\emph{and} the solution size (see \cite{CC20}). These two hardness results illustrate that \hkcut differs significantly from \gkcut in complexity.

\hkcut for fixed $k$ was recently shown to be polynomial-time solvable via a randomized algorithm \cite{CXY19, FPZ19}. The analysis of the randomized algorithm also showed that the number of \minkcutsets is $O(n^{2k-2})$ and they can all be enumerated in randomized polynomial time. A deterministic polynomial-time algorithm was given by Chandrasekaran and Chekuri \cite{CC20}. Subsequently, a deterministic polynomial-time algorithm for \enumhkcut for fixed $k$ was given by Beideman, Chandrasekaran, and Wang \cite{BCW22}.

\mypara{\mmg and \enumMMg.} When $k$ is part of input, \mmg is NP-hard \cite{CW21} while its approximability is open---we do not yet know if it admits a constant factor approximation. When parameterized by $k$, it is $W[1]$-hard and admits a fixed-parameter approximation scheme \cite{CW21}. 

\mmg for fixed $k$ is polynomial-time solvable via the following observation (see \cite{CW21, CC21}): in connected graphs, an optimum $k$-partition for \mmg is a $k$-approximate solution to \gkcut. The randomized algorithm of Karger and Stein implies that the number of $k$-approximate solutions to \gkcut is $n^{O(k^2)}$ and they can all be enumerated in polynomial time \cite{KS96, CQX20, GLL20-STOC, GHLL20}. These two facts together imply that \mmg can be solved in time $n^{O(k^2)}$ and moreover, the number of optimum $k$-partitions for \mmg in a connected graph is $n^{O(k^2)}$ and they can all be enumerated in polynomial time for constant $k$. 

\mypara{\mmh and \enumMMh.} \mmh was discussed as early as 1973 by Lawler \cite{La73}. When $k$ is part of input, \mmh is at least as hard as the densest $k$-subgraph problem (this follows from the reduction in \cite{CL20} and was observed by \cite{CZ-private}). 
For fixed $k$, the approach for \mmg described above does not extend to \mmh. This is because, the number of $k$-approximate solutions to \hkcut can be exponential and hence, they cannot be enumerated efficiently (e.g., we have already seen that the number of $2$-approximate minimum cut-sets in a hypergraph can be exponential). Chandrasekaran and Chekuri \cite{CC21} gave a deterministic polynomial-time algorithm for the more general problem of \mmsymsubmodkpart for fixed $k$ which in turn, implies that \mmh is also solvable efficiently for fixed $k$. Their algorithm finds an optimum $k$-partition for \mmh and is not conducive to enumerate all \minmaxkcutsets. We emphasize that no polynomial bound on the number of \minmaxkcutsets for fixed $k$ was known prior to our work. 

\medskip
For a detailed discussion on other hypergraph $k$-partitioning problems that are special cases of \submodkpart, we  refer the reader to \cite{ZNI05, Zhthesis, OFN12}.

\subsection{Preliminaries} \label{sec:prelims}
Let $G=(V,E)$ be a hypergraph. Throughout, we will follow the notation mentioned in the second paragraph of Section \ref{sec:techniques}. For disjoint $A,B\subseteq V$, we define $E(A,B):=\{e\in E:e\subseteq A\cup B,e\cap A\neq\emptyset, e\cap B\neq\emptyset\}$, and $E[A]:=\{e\in E:e\subseteq A\}$.
We will repeatedly rely on the fact that the hypergraph cut function $d:2^V\rightarrow \R_+$ is symmetric and submodular. We recall that a set function $f:2^V\rightarrow \R$ is \emph{symmetric} if $f(U)=f(\complement{U})$ for all subsets $U\subseteq V$ and is \emph{submodular} if $f(A) + f(B) \ge f(A\cap B)+f(A\cup B)$ for all subsets $A, B\subseteq V$. 

We will need the following partition uncrossing theorem that was proved in previous works on \hkcut and \enumhkcut (see Figure \ref{figure:uncrossing} for an illustration of the sets that appear in the statement of Theorem \ref{theorem:hypergraph-uncrossing}): 

\begin{restatable}{theorem}{thmHypergraphUncrossing}
\label{theorem:hypergraph-uncrossing} \cite{CC20, BCW22}
Let $G=(V,E)$ be a hypergraph, $k\ge 2$ be an integer and
$\emptyset\neq R\subsetneq U\subsetneq V$. Let
$S=\{u_1,\ldots, u_p\}\subseteq U\setminus R$ for $p\ge 2k-2$. Let
$(\complement{A_i}, A_i)$ be a minimum
$((S\cup R)\setminus \set{u_i}, \complement{U})$-terminal cut. Suppose
that $u_i\in A_i\setminus (\cup_{j\in [p]\setminus \set{i}}A_j)$ for
every $i\in [p]$.  Then, the following two hold: 
\begin{enumerate}
    \item There exists a $k$-partition
$(P_1, \ldots, P_k)$ of $V$ with $\complement{U}\subsetneq P_k$ such
that
\[
|\delta(P_1, \ldots, P_k)| \le \frac{1}{2}\min\{\deltacard(A_i) + \deltacard(A_j): i, j\in [p], i\neq j\}.
\]
\item 
Moreover, if there exists a hyperedge 
$e\in E$ such that 
$e$ intersects $W:=\cup_{1\le i<j\le p}(A_i\cap A_j)$, $e$ intersects $Z:=\cap_{i\in [p]} \complement{A_i}$, and $e$ is contained in $W\cup Z$, 
then the inequality 
in the previous conclusion 
is strict.
\end{enumerate}
\end{restatable}

\begin{figure}[htb]
\centering
\includegraphics[width=0.6\textwidth]{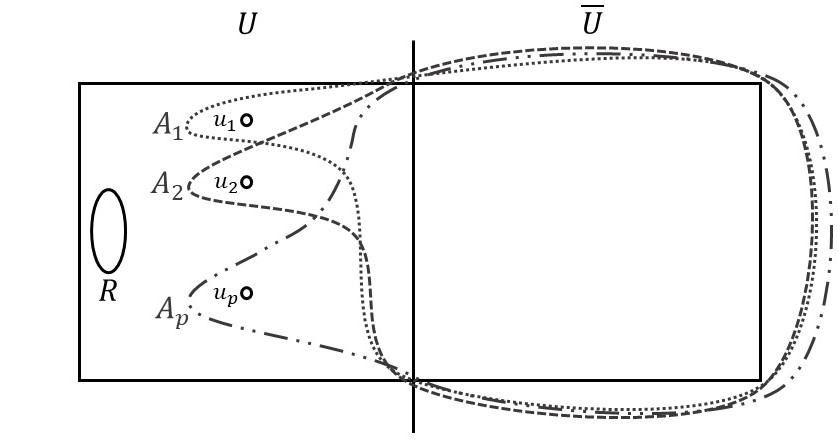}
\caption{Illustration of the sets that appear in the statement of Theorem \ref{theorem:hypergraph-uncrossing}.}
\label{figure:uncrossing}
\end{figure}

\section{A special case of Theorem \ref{thm:cut-set-recovery}}
\label{sec:cut-set-recovery-part-1}

The following is the main theorem of this section. 
Theorem \ref{thm:cut-set-recovery-part-1} implies Theorem \ref{thm:cut-set-recovery} in the special case where the $2$-partition $(U, \complement{U})$ of interest to Theorem \ref{thm:cut-set-recovery} is such that $|\complement{U}|\le 2k-1$. 

\begin{theorem}\label{thm:cut-set-recovery-part-1}
Let $G=(V,E)$ be a hypergraph and let $\opt$ be the optimum value of \hkcut in $G$ for some integer $k\ge 2$. Suppose $(U,\complement{U})$ is a $2$-partition of $V$ with $d(U) = \opt$.
Then, there exists a set $S \subseteq U$ with $|S| \leq 2k-1$ such that every minimum $(S,\overline{U})$-terminal cut $(A, \overline{A})$ satisfies $\delta(A) = \delta(U)$.
\end{theorem}

\begin{proof}

Consider the collection
\[\mathcal{C} := \{ Q \subseteq V \colon \ \overline{U} \subsetneq Q,\ d(Q) \leq d(U), \text{ and } \delta(Q) \neq \delta(U) \}.\]
Let $S$ be an inclusion-wise minimal subset of $U$ such that $S\cap Q\neq \emptyset$ for all $Q\in \mathcal{C}$,  i.e., the set $S$ is completely contained in $U$ and is a minimal transversal of the collection $\mathcal{C}$. Proposition \ref{prop:div-and-conq-dA-equal-dU} and Lemma  \ref{lem:cut-set-recovery-part-one-S-is-small} complete the proof of Theorem  \ref{thm:cut-set-recovery-part-1} for this choice of $S$. 
\end{proof}


\begin{proposition}\label{prop:div-and-conq-dA-equal-dU}
Every minimum $(S, \overline{U})$-terminal cut $(A, \overline{A})$ has $\delta(A) = \delta(U)$.
\end{proposition}

\begin{proof}
Let $(A, \complement{A})$ be a minimum $(S, \complement{U})$-terminal cut. If $A=U$, then we are done, so we may assume that $A\neq U$. This implies that $S\subseteq A$ and $\complement{U}\subsetneq \complement{A}$. 
Since $(U, \overline{U})$ is a $(S, \overline{U})$-terminal cut, we have that $d(\complement{A})=d(A) \leq d(U)$. 
Since $S$ intersects every set in the collection $\collection$, we have that $\complement{A}\not\in \collection$. Hence, $\delta(\complement{A})=\delta(U)$, and by symmetry of cut-sets, $\delta(A)=\delta(U)$. 

\end{proof}


\begin{lemma}\label{lem:cut-set-recovery-part-one-S-is-small}
The size of the subset $S$ is at most $2k-1$.
\end{lemma}
\begin{proof}

For the sake of contradiction, suppose $|S|\geq 2k$. Our proof strategy is to show the existence of a $k$-partition with fewer crossing hyperedges than $\opt$, thus contradicting the definition of $\opt$. 
Let $S:=\{u_1,u_2,\ldots,u_p\}$ for some $p\geq 2k$. For notational convenience, we will use $S - u_i$ to denote $S \setminus \{u_i\}$ and $S - u_i - u_j$ to denote $S \setminus \{u_i, u_j\}$. 
For a subset $X \subseteq U$, we denote the source minimal minimum $(X,\overline{U})$-terminal cut by $(H_X,\overline{H_X})$.


Our strategy to arrive at a $k$-partition with fewer crossing hyperedges than $\opt$ is to apply the second conclusion of Theorem \ref{theorem:hypergraph-uncrossing}. The next few claims will set us up to obtain sets that satisfy the hypothesis of Theorem \ref{theorem:hypergraph-uncrossing}.

\begin{claim}\label{clm:div-and-conq-HS-in-C}
For every $i \in [p]$, we have $\overline{H_{S - u_i}} \in \mathcal{C}$.
\end{claim}

\begin{proof}
Let $i\in [p]$. Since $S$ is a minimal transversal of the collection $\collection$, there exists a set $B_i\in \collection$ such that $B_i\cap S=\set{u_i}$. Hence, $(\complement{B_i}, B_i)$ is a $(S-u_i, \complement{U})$-terminal cut. Therefore, 
\[
d(\overline{H_{S - u_i}}) \leq d(B_i) \leq d(U). 
\]
Since $(H_{S - u_i}, \overline{H_{S-u_i}})$ is a $(S - u_i, \overline{U})$-terminal cut, we have that $\overline{U} \subseteq \overline{H_{S - u_i}}$. If $d(\overline{H_{S-u_i}}) < d(U)$, then $\delta(\overline{H_{S - u_i}}) \neq \delta(U)$ and  $\complement{U}\subsetneq \complement{H_{S-u_i}}$, and consequently, $\overline{H_{S - u_i}} \in \mathcal{C}$. So, we will assume henceforth that $d(\overline{H_{S - u_i}})= d(U)$. 

Since $(H_{S - u_i} \cap \overline{B_i}, \overline{H_{S - u_i} \cap \overline{B_i}})$ is a $(S - u_i, \overline{U})$-terminal cut, we have that 
\[d(H_{S-u_i} \cap \overline{B_i}) \geq d(H_{S - u_i}).\]
Since  $(H_{S - u_i} \cup \overline{B_i}, \overline{H_{S - u_i} \cup \overline{B_i}})$ is a $(S-u_i, \overline{U})$-terminal cut, we have that 
\[d(H_{S - u_i} \cup \overline{B_i}) \geq d(H_{S - u_i}).\]
Therefore, by submodularity of the hypergraph cut function, we have that 
\begin{equation}\label{eqn:div-and-conq-HS-in-C-3}
2d(U) \geq d(H_{S-u_i}) + d(B_i) \geq d(H_{S-u_i} \cap \overline{B_i}) + d(H_{S - u_i} \cup \overline{B_i}) \geq 2d(H_{S - u_i}) = 2d(U).
\end{equation}
Therefore, all inequalities above should be equations. In particular, we have that $d(H_{S-u_i} \cap \overline{B_i}) = d(U) = d(B_i)=d(H_{S - u_i})$ and hence, $(H_{S - u_i} \cap \overline{B_i}, \overline{H_{S - u_i} \cap \overline{B_i}})$ is a minimum $(S - u_i, \overline{U})$-terminal cut. Since $(H_{S-u_i}, \overline{H_{S-u_i}})$ is a source minimal minimum $(S-u_i, \overline{U})$-terminal cut, we must have $H_{S-u_i} \cap \overline{B_i} = H_{S - u_i}$, and thus $H_{S - u_i} \subseteq \overline{B_i}$. Therefore, $B_i \subseteq \overline{H_{S-u_i}}$. 
Since $B_i \in \mathcal{C}$, we have $\delta(B_i) \neq \delta(U)$. However, $d(B_i)=d(U)$. Therefore $\delta(U) \setminus \delta(B_i) \neq \emptyset$. Let $e \in \delta(U) \setminus \delta(B_i)$. Since $e \in \delta(\overline{U})$, but $e \not\in \delta(B_i)$, and $\overline{U} \subseteq B_i$, we have that $e \subseteq B_i$, and thus $e \subseteq \overline{H_{S - u_i}}$. Thus, we conclude that $\delta(U) \setminus \delta(\overline{H_{S - u_i}}) \neq \emptyset$, and so $\delta(\overline{H_S - u_i}) \neq \delta(U)$. This also implies that $\complement{U}\subsetneq \complement{H_{S-u_i}}$. Thus, $\overline{H_{S - u_i}} \in \mathcal{C}$.
\end{proof}

Claim \ref{clm:div-and-conq-HS-in-C} implies the following Corollary.

\begin{corollary}\label{cor:div-and-conq-ui-in-Ai}
For every $i \in [p]$, we have $u_i \in \overline{H_{S-u_i}}$.
\end{corollary}
\begin{proof}
By definition, $S- u_i \subseteq H_{S-u_i}$, so $S \cap \overline{H_{S - u_i}} \subseteq \{u_i\}$. By Claim \ref{clm:div-and-conq-HS-in-C} we have that $\overline{H_{S - u_i}} \in \mathcal{C}$. Since $S$ is a transversal of the collection $\mathcal{C}$, we have that $S \cap \overline{H_{S - u_i}} \neq \emptyset$. So, the vertex $u_i$ must be in $\overline{H_{S - u_i}}$.
\end{proof}




Having obtained Corollary \ref{cor:div-and-conq-ui-in-Ai}, the next few claims (Claims \ref{clm:div-and-conq-HSij-sub-HSi}, \ref{clm:div-and-conq-eqd}, \ref{clm:div-and-conq-Ci_cap_cup_eq_V1}, and \ref{clm:div-and-conq-HSl-subset-HSi-cup-HSj}) are similar to the claims appearing in the proof of a structural theorem that appeared in \cite{BCW22}. Since the hypothesis of the structural theorem that we are proving here is different from theirs, we present the complete proofs of these claims here. 
 The way in which we use the claims will also be different from \cite{BCW22}. 

The following claim will help in showing that $u_i, u_j \not\in H_{S-u_i-u_j}$, which in turn, will be used to show that the hypothesis of Theorem \ref{theorem:hypergraph-uncrossing} is satisfied by suitably chosen sets. 

\begin{claim}\label{clm:div-and-conq-HSij-sub-HSi}
For every $i,j \in [p]$, we have $H_{S - u_i - u_j} \subseteq H_{S - u_i}$.
\end{claim}

\begin{proof}
We may assume that $i\neq j$. 
We note that $(H_{S - u_i-u_j} \cap H_{S - u_i}, \overline{ H_{S - u_i-u_j} \cap H_{S - u_i } })$ is a $(S - u_i-u_j, \overline{U})$-terminal cut. Therefore, 
\begin{equation}\label{eqn:div-and-conq-HSij-sub-HSi-1}
    d(H_{S - u_i-u_j} \cap H_{S - u_i}) \geq d(H_{S - u_i-u_j}).
\end{equation}
Also, $(H_{S - u_i-u_j} \cup H_{S - u_i}, \overline{ H_{S - u_i-u_j} \cup H_{S - u_i } })$ is a $(S - u_i, \overline{U})$-terminal cut. Therefore, 
\begin{equation}\label{eqn:div-and-conq-HSij-sub-HSi-2}
    d(H_{S - u_i-u_j} \cup H_{S - u_i}) \geq d(H_{S - u_i}).
\end{equation}
By submodularity of the hypergraph cut function and inequalities (\ref{eqn:div-and-conq-HSij-sub-HSi-1}) and (\ref{eqn:div-and-conq-HSij-sub-HSi-2}), we have that 
\begin{align*}
d(H_{S - u_i-u_j}) + d(H_{S - u_i}) &\geq d(H_{S - u_i-u_j} \cap H_{S - u_i}) + d(H_{S - u_i-u_j} \cup H_{S - u_i}) \\ 
&\geq d(H_{S - u_i-u_j}) + d(H_{S - u_i}).
\end{align*}

Therefore, inequality (\ref{eqn:div-and-conq-HSij-sub-HSi-1}) is an equation, and consequently, $(H_{S - u_i-u_j} \cap H_{S - u_i}, \overline{ H_{S - u_i-u_j} \cap H_{S - u_i } })$ is a minimum $(S - u_i-u_j, \overline{U})$-terminal cut. 
If $H_{S - u_i-u_j} \setminus H_{S - u_i} \neq \emptyset$, then \\ $(H_{S - u_i-u_j} \cap H_{S - u_i}, \complement{H_{S - u_i-u_j} \cap H_{S - u_i}})$ contradicts source minimality of the minimum $(S - u_i-u_j,\overline{U})$-terminal cut $(H_{S - u_i-u_j}, \complement{H_{S - u_i-u_j}})$. Hence, $H_{S - u_i-u_j} \setminus H_{S - u_i}=\emptyset$ and consequently, $H_{S - u_i-u_j} \subseteq H_{S - u_i}$.
\end{proof}

Claim \ref{clm:div-and-conq-HSij-sub-HSi} implies the following Corollary.
\begin{corollary}\label{cor:div-and-conq-ij-not-in-HSij}
For every $i,j \in [p]$, we have $u_i,u_j \not\in H_{S - u_i - u_j}$.
\end{corollary}

\begin{proof}
By Corollary \ref{cor:div-and-conq-ui-in-Ai}, we have that $u_i \not\in H_{S - u_i}$. Therefore, $u_i, u_j \not\in H_{S - u_i} \cap H_{S - u_j}$. By Claim \ref{clm:div-and-conq-HSij-sub-HSi}, $H_{S -u_i- u_j} \subseteq H_{S - u_i}$ and $H_{S -u_i- u_j} \subseteq H_{S - u_j}$. Therefore, $H_{S -u_i- u_j} \subseteq H_{S - u_i} \cap H_{S - u_j}$, and thus, $u_i,u_j \not\in H_{S - u_i-u_j}$.
\end{proof}

The next claim will help in controlling the cost of the $k$-partition that we will obtain by applying Theorem \ref{theorem:hypergraph-uncrossing}. 

\begin{claim}\label{clm:div-and-conq-eqd}
For every $i,j \in [p]$, we have $d(H_{S - u_i}) = d(U) = d(H_{S - u_i -u_j})$.
\end{claim}

\begin{proof}
Let $a,b\in [p]$. We will show that $d(H_{S-u_a})=d(U)=d(H_{S-u_a-u_b})$. 
Since $(U, \overline{U})$ is a $(S - u_a,\overline{U})$-terminal cut, we have that $d(H_{S - u_a}) \leq d(U)$.
Since $(H_{S - u_a}, \overline{H_{S - u_a}})$ is a $(S - u_a-u_b,\overline{U})$-terminal cut, we have that $d(H_{S - u_a-u_b}) \leq d(H_{S - u_a}) \leq d(U)$. Thus, in order to prove the claim, it suffices to show that $d(H_{S - u_a-u_b}) \geq d(U)$.

Suppose for contradiction that $d(H_{S - u_a-u_b}) < d(U)$. Let $\ell \in [p] \setminus \{a,b\}$ be an arbitrary element (which exists since we have assumed that $p \geq 2k$ and $k \geq 2$). Let $R:= \{u_{\ell}\}$, $S':=S - u_a - u_{\ell}$ 
, and $A_i := \overline{H_{S -u_a-u_i}}$ for every $i \in [p] \setminus \{a,\ell\}$. 
We note that $|S'| = p-2 \geq 2k-2$. 
By definition, $(\complement{A_i}, A_i)$ is a minimum $(S -u_a -u_i, \overline{U})$-terminal cut for every $i \in [p] \setminus \{a,\ell\}$. 
Moreover, 
by Corollary \ref{cor:div-and-conq-ij-not-in-HSij}, 
we have that 
$u_i\in A_i\setminus (\cup_{j\in [p]\setminus \{a, i, \ell\}}A_j)$ 
for every $i \in [p] \setminus \{a,\ell\}$. 
Hence, the sets $U$, $R$, and $S'$, and the cuts $(\complement{A_i}, A_i)$ for $i\in [p]\setminus \set{a,\ell}$ satisfy the conditions of Theorem \ref{theorem:hypergraph-uncrossing}. Therefore, by the first conclusion of Theorem \ref{theorem:hypergraph-uncrossing}, there exists a $k$-partition $\mathcal{P}'$ with 
\[|\delta(\mathcal{P}')| \leq \frac{1}{2}\min\{d(H_{S -u_a- u_i})+d(H_{S -u_a-u_j}) \colon i,j \in [p] \setminus \{a,\ell\} \}.\]
By assumption, $d(H_{S - u_a-u_b}) < d(U)$ and $b\in [p]\setminus \{a, \ell\}$, so $\min\{d(H_{S -u_a-u_i})\colon i \in [p] \setminus \{a,\ell\} \} < d(U)$. Since $(U, \overline{U})$ is a $(S - u_a-u_i, \overline{U})$-terminal cut, we have that $d(H_{S -u_a-u_i}) \leq d(U)$ for every $i \in [p] \setminus \{a,\ell\}$. Therefore,
\[\frac{1}{2}\min\{d(H_{S -u_a- u_i})+d(H_{S -u_a-u_j}) \colon i,j \in [p] \setminus \{a,\ell\} \} < d(U) = \opt.\]
Thus, we have that $|\delta(\mathcal{P'})| < \opt$, which is a contradiction. 
\end{proof}

The next two claims will help in arguing the existence of a hyperedge satisfying the conditions of the second conclusion of Theorem \ref{theorem:hypergraph-uncrossing}. In particular, we will need Claim \ref{clm:div-and-conq-HSl-subset-HSi-cup-HSj}. The following claim will help in proving Claim \ref{clm:div-and-conq-HSl-subset-HSi-cup-HSj}.
\begin{claim}\label{clm:div-and-conq-Ci_cap_cup_eq_V1}
For every $i,j \in [p]$, we have 
\[d(H_{S - u_i} \cap H_{S - u_j}) = d(U) = d(H_{S - u_i} \cup H_{S - u_j}).\]
\end{claim}
\begin{proof}
Since $(H_{S - u_i} \cap H_{S - u_j}, \overline{H_{S - u_i} \cap H_{S - u_j}} )$ is a $(S - u_i-u_j, \overline{U})$-terminal cut, we have that $d(H_{S - u_i} \cap H_{S - u_j}) \geq d(H_{S - u_i-u_j})$. By Claim \ref{clm:div-and-conq-eqd}, we have that $d(H_{S - u_i-u_j}) = d(U) = d(H_{S - u_i})$. Therefore,
\begin{equation}\label{eq:div-and-conq-Ci_cap_cup_eq_V1_eq1}
    d(H_{S - u_i} \cap H_{S - u_j}) \geq d(H_{S - u_i}).
\end{equation}
Since $(H_{S - u_i} \cup H_{S - u_j}, \overline{H_{S - u_i} \cup H_{S - u_j}})$ is a $(S - u_j, \overline{U})$-terminal cut, we have that
\begin{equation}\label{eq:div-and-conq-Ci_cap_cup_eq_V1_eq2}
    d(H_{S - u_i} \cup H_{S - u_j}) \geq d(H_{S - u_j}).
\end{equation}
By submodularity of the hypergraph cut function and inequalities (\ref{eq:div-and-conq-Ci_cap_cup_eq_V1_eq1}) and (\ref{eq:div-and-conq-Ci_cap_cup_eq_V1_eq2}), we have that
\[
d(H_{S - u_i}) + d(H_{S - u_j}) \geq d(H_{S - u_i} \cap H_{S - u_j}) + d(H_{S - u_i} \cup H_{S - u_j}) \geq d(H_{S - u_i}) + d(H_{S - u_j}).
\]
Therefore, inequalities (\ref{eq:div-and-conq-Ci_cap_cup_eq_V1_eq1}) and (\ref{eq:div-and-conq-Ci_cap_cup_eq_V1_eq2}) are equations. Thus, by Claim \ref{clm:div-and-conq-eqd}, we have that
\[
d(H_{S - u_i} \cap H_{S - u_j}) = d(H_{S - u_i}) = d(U),
\]
and
\[
d(H_{S - u_i} \cup H_{S - u_j}) = d(H_{S - u_j}) = d(U).
\]
\end{proof}

\begin{claim}\label{clm:div-and-conq-HSl-subset-HSi-cup-HSj}
For every $i, j, \ell \in [p]$ with $i \neq j$, we have $H_{S - u_{\ell}} \subseteq H_{S - u_i} \cup H_{S - u_j}$.
\end{claim}

\begin{proof}
If $\ell = i$ or $\ell = j$ the claim is immediate. Thus, we assume that $\ell\not\in \{i, j\}$.
Let $Q := H_{S -u_{\ell}} \setminus (H_{S - u_i} \cup H_{S - u_j})$. We need to show that $Q = \emptyset$. We will show that $(H_{S - u_{\ell}} \setminus Q, \overline{H_{S - u_{\ell}}\setminus Q})$ is a minimum $(S - u_{\ell}, \overline{U})$-terminal cut. Consequently, $Q$ must be empty 
(otherwise, $H_{S - u_{\ell}} \setminus Q \subsetneq H_{S - u_{\ell}}$ and hence, $(H_{S - u_{\ell}} \setminus Q, \overline{H_{S - u_{\ell}}\setminus Q})$ contradicts source minimality of the minimum $(S - u_{\ell}, \overline{U})$-terminal cut $(H_{S-u_{\ell}}, \complement{H_{S-u_{\ell}}})$). 

We now show that $(H_{S - u_{\ell}} \setminus Q, \overline{H_{S - u_{\ell}}\setminus Q})$ is a minimum $(S - u_{\ell}, \overline{U})$-terminal cut. 
Since $H_{S - u_{\ell}} \setminus Q = H_{S - u_{\ell}} \cap (H_{S - u_i} \cup H_{S - u_j})$, we have that $S -u_i-u_j-u_{\ell} \subseteq H_{S - u_{\ell}} \setminus Q$. We also know that $u_i$ and $u_j$ are contained in both $H_{S - u_{\ell}}$ and $H_{S - u_i} \cup H_{S - u_j}$. Therefore, $S - u_{\ell} \subseteq H_{S - u_{\ell}} \setminus Q$. Thus, $(H_{S - u_{\ell}} \setminus Q, \overline{H_{S - u_{\ell}} \setminus Q})$ is a $(S - u_{\ell}, \overline{U})$-terminal cut. Therefore, 
\begin{equation}\label{eqn:div-and-conq-Cl_subset_CiCj_eq1}
  d(H_{S - u_{\ell}} \cap (H_{S - u_i} \cup H_{S - u_j})) = d(H_{S - u_{\ell}} \setminus Q) \geq d(H_{S - u_{\ell}}).  
\end{equation}
We also have that $(H_{S - u_{\ell}} \cup (H_{S - u_i} \cup H_{S - u_j}), \overline{H_{S - u_{\ell}} \cup (H_{S - u_i} \cup H_{S - u_j})})$ is a $(S - u_i, \overline{U})$-terminal cut. Therefore, $d(H_{S - u_{\ell}} \cup (H_{S - u_i} \cup H_{S - u_j})) \geq d(H_{S - u_i})$. By Claims \ref{clm:div-and-conq-eqd} and \ref{clm:div-and-conq-Ci_cap_cup_eq_V1}, we have that $d(H_{S - u_i}) = d(V_1)=d(H_{S - u_i} \cup H_{S - u_j})$. Therefore, 
\begin{equation}\label{eqn:div-and-conq-Cl_subset_CiCj_eq2}
    d(H_{S - u_{\ell}} \cup (H_{S - u_i} \cup H_{S - u_j})) \geq  d(H_{S - u_i} \cup H_{S - u_j}).
\end{equation}
By submodularity of the hypergraph cut function and inequalities (\ref{eqn:div-and-conq-Cl_subset_CiCj_eq1}) and (\ref{eqn:div-and-conq-Cl_subset_CiCj_eq2}), we have that
\begin{align*}
d(H_{S - u_{\ell}}) + d(H_{S - u_i} \cup H_{S - u_j}) 
&\geq d(H_{S - u_{\ell}} \cap (H_{S - u_i} \cup H_{S - u_j})) +  d(H_{S - u_{\ell}} \cup (H_{S - u_i} \cup H_{S - u_j})) \\
&\geq d(H_{S - u_{\ell}}) + d(H_{S - u_i} \cup H_{S - u_j}).
\end{align*}
Therefore, inequalities (\ref{eqn:div-and-conq-Cl_subset_CiCj_eq1}) and (\ref{eqn:div-and-conq-Cl_subset_CiCj_eq2}) are equations, so $(H_{S - u_{\ell}} \setminus Q, \overline{H_{S - u_{\ell}}\setminus Q})$ is a minimum $(S - u_{\ell}, \overline{U})$-terminal cut.
\end{proof}

Let $R:= \{u_{p}\}$, $S' := S - u_p$,  and $(\overline{A_i}, A_i) := (H_{S - u_i}, \overline{H_{S - u_i}})$ for every $i \in [p-1]$. By definition, $(\complement{A_i}, A_i)$ is a minimum $(S - u_i, \overline{U})$-terminal cut for every $i \in [p-1]$. Moreover, by Corollary \ref{cor:div-and-conq-ui-in-Ai}, we have that $u_i\in A_i\setminus (\cup_{j\in [p-1]\setminus \{i\}}A_j)$. Hence, the sets $U$, $R$, and $S'$, and the cuts $(\complement{A_i}, A_i)$ for $i\in [p-1]$ satisfy the conditions of Theorem \ref{theorem:hypergraph-uncrossing}. We will use the second conclusion of Theorem \ref{theorem:hypergraph-uncrossing}. 
We now show that there exists a hyperedge satisfying the conditions mentioned in the second conclusion of Theorem \ref{theorem:hypergraph-uncrossing}. We will use Claim \ref{clm:div-and-conq-existence-of-e} below to prove this. 
Let $W:=\cup_{1\le i<j\le p-1}(A_i\cap A_j)$ and $Z:=\cap_{i\in [p-1]} \complement{A_i}$ as in the statement of Theorem \ref{theorem:hypergraph-uncrossing}.

\begin{claim}\label{clm:div-and-conq-existence-of-e}
There exists a hyperedge $e \in E$ such that $e \cap W \neq \emptyset$, $e \cap Z \neq \emptyset$, and $e \subseteq W \cup Z$.
\end{claim}
\begin{proof}

We note that $S \subseteq (S - u_i) \cup (S - u_j) \subseteq H_{S - u_i} \cup H_{S - u_j}$ for every distinct $i, j \in [p-1]$. Therefore, $S \cap (A_i \cap A_j) = \emptyset$ for every distinct $i,j \in [p-1]$, and thus $S \cap W = \emptyset$.  Since $S$ is a transversal of the collection $\collection$, it follows that the set $W$ is not in the collection $\collection$. 

By definition, $\overline{U} \subseteq A_i$ for every $i \in [p-1]$, and thus $\overline{U} \subseteq W$. Since $W \not\in \mathcal{C}$, either 
$d(W) > d(U)$ or $\delta(W) = \delta(U)$.  
By Claim \ref{clm:div-and-conq-HS-in-C}, we have that $\overline{H_{S-u_{p}}} \in \mathcal{C}$, and thus,  $d(\overline{H_{S-u_{p}}}) \leq d(U)$ and $\delta(\overline{H_{S-u_{p}}}) \neq \delta(U)$. Consequently, $d(W) \geq d(\overline{H_{S-u_{p}}})$, and $\delta(W) \neq \delta(\overline{H_{S-u_{p}}})$, and thus, $\delta(W) \setminus \delta(\overline{H_{S-u_{p}}}) \neq \emptyset$. 
Let $e \in \delta(W) \setminus \delta(\overline{H_{S-u_{p}}})$. We will show that this choice of $e$ achieves the desired properties.
    
For each $i \in [p]$, let $Y_i := \overline{H_{S - u_{i}}} \setminus W$. By Claim \ref{clm:div-and-conq-HSl-subset-HSi-cup-HSj}, for every $i, j, \ell \in [p]$ with $i \neq j$ we have that $H_{S - u_{\ell}} \subseteq H_{S - u_i} \cup H_{S - u_j}$. Therefore $\overline{H_{S-u_i}} \cap \overline{H_{S-u_j}} \subseteq \overline{H_{S-u_{\ell}}}$ for every such $i, j, \ell\in [p]$, and hence $W \subseteq \overline{H_{S - u_{\ell}}}$ for every $\ell \in [p]$. Thus, $W \subseteq \overline{H_{S-u_{p}}}$. Since $e \in \delta(W) \setminus \delta(\overline{H_{S-u_{p}}})$, we have that $e \subseteq W \cup Y_{p}$, $e \cap W \neq \emptyset$ and $e \cap Y_{p} \neq \emptyset$. Therefore, in order to show that $e$ has the three desired properties as in the claim, it suffices to show that $Y_{p} \subseteq Z$. We prove this next.

By definition, $Y_{p} \cap W = \emptyset$. By Claim \ref{clm:div-and-conq-HSl-subset-HSi-cup-HSj}, for every $i \in [p-1]$, we have that $\overline{H_{S-u_{p}}} \cap \overline{H_{S-u_i}} \subseteq \overline{H_{S-u_1}}$ and $\overline{H_{S-u_{p}}} \cap \overline{H_{S-u_i}} \subseteq \overline{H_{S-u_2}}$, so $\overline{H_{S-u_{p}}} \cap \overline{H_{S-u_i}} \subseteq \overline{H_{S-u_1}} \cap \overline{H_{S-u_2}} \subseteq W$. Thus, for every $i \in [p-1]$, $Y_{p} \cap Y_i \subseteq \overline{H_{S-u_{p}}} \cap \overline{H_{S-u_i}} \subseteq W$, so since $Y_{p} \cap W = \emptyset$, we have that $Y_{p} \cap Y_i = \emptyset$ for every $i \in [p-1]$. Therefore, 
\[
Y_{p} \subseteq \overline{W \cup \left( \bigcup_{i=1}^{p-1} Y_i \right)} = \overline{\bigcup_{i=1}^{p-1} \overline{H_{S-u_i}}} = \bigcap_{i=1}^{p-1} H_{S-u_i} = Z.
\]




\end{proof}

By Claim \ref{clm:div-and-conq-existence-of-e}, there is a hyperedge $e$ satisfying the conditions of the second conclusion of Theorem \ref{theorem:hypergraph-uncrossing}. Therefore, by Theorem \ref{theorem:hypergraph-uncrossing}, there exists a $k$-partition $\mathcal{P'}$ with 
\begin{align*}
|\delta(\mathcal{P'})| 
&< \frac{1}{2}\min\{\deltacard(A_i) + \deltacard(A_j): i, j\in [p-1], i\neq j\}\\
&= d(U) \quad \quad \quad \quad \quad  \quad \text{(By Claim \ref{clm:div-and-conq-eqd})}\\
&= \opt. \quad \quad \quad \quad \text{(By assumption of the theorem)}
\end{align*}
Thus, we have obtained a $k$-partition $\mathcal{P}'$ with $|\delta(\mathcal{P}')| < \opt$, which is a contradiction.

\end{proof}

\section{Proof of Theorem \ref{thm:cut-set-recovery}}\label{sec:cut-set-recovery}


We prove Theorem \ref{thm:cut-set-recovery} in this section. Applying Theorem \ref{thm:cut-set-recovery-part-1} to $(\overline{U}, U)$ yields the following corollary.

\begin{corollary}\label{cor:cut-set-recovery-part-two}
Let $G=(V,E)$ be a hypergraph and let $\opt$ be the optimum value of \hkcut in $G$ for some integer $k\ge 2$. Suppose $(U,\complement{U})$ is a $2$-partition of $V$ with $d(U) = \opt$.
Then, there exists a set $T \subseteq \overline{U}$ with $|T| \leq 2k-1$ such that every minimum $(U,T)$-terminal cut $(A, \overline{A})$ satisfies $\delta(A) = \delta(U)$.
\end{corollary}

We now restate Theorem \ref{thm:cut-set-recovery} and prove it using Theorem \ref{thm:cut-set-recovery-part-1} and Corollary \ref{cor:cut-set-recovery-part-two}.

\thmStrongerStructure*
\begin{proof} 
By Theorem \ref{thm:cut-set-recovery-part-1}, there exists a subset $S \subseteq U$ with $|S|\le 2k-1$ such that every minimum $(S, \overline{U})$-terminal cut $(A, \overline{A})$ has $\delta(A) = \delta(U)$. By Corollary \ref{cor:cut-set-recovery-part-two}, there exists a subset $T \subseteq \overline{U}$ with $|T|\le 2k-1$ such that every minimum $(U,T)$-terminal cut $(A, \overline{A})$ has $\delta(A) = \delta(U)$. We will show that every minimum $(S,T)$-terminal cut $(A, \overline{A})$ has $\delta(A) = \delta(U)$. We will need the following claim.



\begin{claim}\label{clm:div-and-conq-dYeqdU}
Let $(Y, \overline{Y})$ be the source minimal minimum $(S,T)$-terminal cut. Then $\delta(Y) = \delta(U)$.
\end{claim}
\begin{proof}
Since $(U, \overline{U})$ is a $(S,T)$-terminal cut, and $(Y, \overline{Y})$ is a minimum $(S,T)$-terminal cut, we have that 
\[
d(U) \geq d(Y). 
\]
Since $(U \cap Y, \overline{U \cap Y})$ is a 
$(S, \overline{U})$-terminal cut, we have that 
\[
d(U \cap Y) \geq d(U). 
\]
Since $(U \cup Y, \overline{U \cup Y})$ is a $(U, T)$-terminal cut, we have that 
\[
d(U \cup Y) \geq d(U). 
\]
Thus, by the submodularity of the hypergraph cut function we have that
\[
2d(U) \geq d(U) + d(Y) \geq d(U \cap Y) + d(U \cup Y) \geq 2d(U).
\]
Therefore, we have that $d(U \cap Y) = d(U)$, so $(U \cap Y, \overline{U \cap Y})$ is a minimum $(S,T)$-terminal cut. Since $(Y, \complement{Y})$ is the source minimal $(S,T)$-terminal cut, we have that $U \cap Y = Y$, and hence $Y \subseteq U$. Therefore, $(Y, \overline{Y})$ is a minimum $(S, \overline{U})$-terminal cut. By the choice of $S$, we have that $\delta(Y) = \delta(U)$.
\end{proof}

Applying Claim \ref{clm:div-and-conq-dYeqdU} to both sides of the partition $(U,\overline{U})$, we have that the source minimal minimum $(S,T)$-terminal cut $(Y, \overline{Y})$ has $\delta(Y) = \delta(U)$, and the source minimal minimum $(T,S)$-terminal cut $(Z,\overline{Z})$ has $\delta(Z) = \delta(U)$. Therefore, for every $e \in \delta(U)$, we have that $e \cap Y \neq \emptyset$ and $e \cap Z \neq \emptyset$.

Let $(A, \overline{A})$ be a minimum $(S,T)$-terminal cut. Since $(Y, \overline{Y})$ is the source minimal minimum $(S,T)$-terminal cut, we have that $Y \subseteq A$. Since $(Z, \overline{Z})$ is the source minimal minimum $(T,S)$-terminal cut, we have that $Z \subseteq \overline{A}$. Since every $e \in \delta(U)$ intersects both $Y$ and $Z$, it follows that every $e \in \delta(U)$ intersects both $A$ and $\overline{A}$, and hence,  $\delta(U) \subseteq \delta(A)$. Since $(A, \overline{A})$ is a minimum $(S,T)$-terminal cut, $d(A) \leq d(U)$, and thus we have that $\delta(A) = \delta(U)$.

\end{proof}

\section{Algorithm for \enumhkcut}\label{sec:enumhkcut-algo}
In this section, we design a 
deterministic algorithm for \enumhkcut that is 
based on
divide and conquer and has a run-time of $n^{O(k)}$ source minimal minimum $(s, t)$-terminal cut computations, where $n$ is the number of vertices in the input hypergraph. 
The high-level idea is to use minimum $(S, T)$-terminal cuts to enumerate a collection of candidate cuts such that 
for every optimum $k$-partition for \hkcut, 
either the union of some $k/2$ parts of the optimum $k$-partition is contained in the candidate collection 
or we find the set of hyperedges crossing this optimum $k$-partition. 
This helps in cutting the recursion depth to $\log{k}$ which saves on overall run-time. 
We describe the algorithm in Figure \ref{Algo: DC enum-cuts} and its  guarantees in Theorem \ref{theorem: DC enum algorithm}. We recall that for a hypergraph $G=(V,E)$ and a subset $A\subseteq V$, the subgraph $G[A]$ induced by $A$ is given by $G[A]=(A, E')$, where $E':=\{e\in E: e\subseteq A\}$. 

Theorem \ref{theorem: DC enum algorithm} is a self-contained proof that the number of \minkcutsets in a $n$-vertex hypergraph is $O(n^{8k\log{k}})$ and the run-time of the algorithm in Figure \ref{Algo: DC enum-cuts} is $O(n^{8k\log{k}})$ source minimal minimum $(s,t)$-terminal cut computations. 
In Lemma \ref{lemma:DC-better-run-time}, we improve the run-time analysis of the same algorithm to $O(n^{16k})$ source minimal minimum $(s,t)$-terminal cut computations. For this, we exploit the known fact that the number of \minkcutsets in a $n$-vertex hypergraph is $O(n^{2k-2})$ (via the randomized algorithm in \cite{CXY19}). 

Theorem \ref{theorem: DC enum algorithm} and Lemma \ref{lemma:DC-better-run-time} together imply Theorem \ref{thm:k-cut-set-enumeration} since the source minimal minimum $(s,t)$-terminal cut in a $n$-vertex hypergraph of size $p$ can be computed in time $O(np)$ \cite{ChekuriX18}. 

\begin{theorem} \label{theorem: DC enum algorithm}
Let $G=(V,E)$ be a $n$-vertex hypergraph of size $p$ and let $k$ be a positive integer. Then, Algorithm Enum-Cut-Sets$(G,k)$ in Figure \ref{Algo: DC enum-cuts} returns the family of all \minkcutsets in $G$ and it can be implemented to run in time $O(n^{(8k-6)\log k})T(n,p)$, where $T(n,p)$ denotes the time complexity for computing the source minimal minimum $(s,t)$-terminal cut in a $n$-vertex hypergraph of size $p$. Moreover, the cardinality of the family returned by the algorithm is $O(n^{(8k-6)\log k})$.
\end{theorem}

\begin{figure*}[ht]
\centering\small
\begin{algorithm}
\textul{Algorithm Enum-Cut-Sets$(G=(V,E),k)$}\+
\\{\bf Input:} Hypergraph $G=(V,E)$ and an integer $k\geq 1$
\\{\bf Output:} Family of all \minkcutsets in $G$
\\ If $k=1$ \+
\\ Return $\{\emptyset\}$\-
\\ Else\+
\\ Initialize $\mathcal{C}\gets\emptyset$, $\mathcal{F}\gets\emptyset$
\\ For each pair $(S,T)$ such that $S,T\subseteq V$ with $S\cap T=\emptyset$ and $|S|,|T|\leq 2k-1$ \+
\\ Compute the source minimal minimum $(S,T)$-terminal cut $(A,\complement{A})$
\\ If $G-\delta(A)$ has at least $k$ connected components\+
\\ $\mathcal{F}\gets\mathcal{F}\cup\{\delta(A)\}$\-
\\ Else\+
\\ $\mathcal{C}\gets\mathcal{C}\cup\{A\}$\-\-
\\ For each $A\in\mathcal{C}$ such that $|A|\geq \lfloor k/2\rfloor$ and $|\complement{A}|\geq k-\lfloor k/2\rfloor$\+
\\ $\mathcal{F}_A\gets$Enum-Cut-Sets$(G[A],\lfloor k/2\rfloor)$
\\ $\mathcal{F}'_A\gets$Enum-Cut-Sets$(G[\complement{A}],k-\lfloor k/2\rfloor)$
\\ $\mathcal{F}\gets\mathcal{F}\cup\{\delta(A)\cup F\cup F':F\in\mathcal{F}_A,F'\in\mathcal{F}'_A\}$\-
\\ Among all $k$-cut-sets in the family $\mathcal{F}$, return the subfamily that are of smallest size
\end{algorithm}
\caption{Divide-and-conquer algorithm to enumerate hypergraph minimum $k$-cut-sets}
\label{Algo: DC enum-cuts}
\end{figure*}


\begin{proof}
We begin by showing correctness. 
The last step of the algorithm considers only $k$-cut-sets in the family $\family$, so the algorithm returns a subfamily of $k$-cut-sets. We only have to show that every \minkcutset is in the family $\family$; 
this will also guarantee that every $k$-cut-set in the returned subfamily is indeed a \minkcutset. We show this by induction on $k$.

For the base case of $k=1$, the only \minkcutset is the empty set which is contained in the returned family. We now show the induction step. Assume that $k\ge 2$. Let $F\subseteq E$ be a \minkcutset in $G$ and let $(V_1,\ldots,V_k)$ be an optimum $k$-partition for \hkcut such that $F=\delta(V_1,\ldots,V_k)$. We will show that $F$ is in the family returned by the algorithm. Let $U:=\cup_{i=1}^{\lfloor k/2\rfloor}V_i$. We distinguish between the following two cases:

\begin{enumerate}
    \item Suppose $d(U)<\opt$.
    
    By Theorem \ref{theorem: structure thm 1}, there exist disjoint subsets $S,T\subseteq V$ with $|S|,|T|\leq 2k-2$ such that $(U,\complement{U})$ is the unique minimum $(S,T)$-terminal cut. 
    Hence, the set $U$ is in the collection $\mathcal{C}$ 
    Moreover, $U$ contains $\lfloor k/2\rfloor$ non-empty sets $V_1,V_2,\ldots,V_{\lfloor k/2\rfloor}$, so we have $|U|\geq \lfloor k/2\rfloor$. Similarly, we have $|\complement{U}|\geq k-\lfloor k/2\rfloor$. Since $(V_1,\ldots,V_k)$ is an optimum $k$-partition for \hkcut, the set $\{e\in F\backslash \delta(U):e\subseteq U\}$ is a \textsc{min-$\lfloor k/2\rfloor$-cut-set} 
    in $G[U]$. Similarly, the set $\{e\in F\backslash \delta(U):e\subseteq \complement{U}\}$ is a 
    \textsc{min-$(k-\lfloor k/2\rfloor)$-cut-set} 
    in $G[\complement{U}]$. Since $d(U)<\opt$, we know that $G-\delta(U)$ has less than $k$ connected components. Therefore, the set $U$ is in the collection $\mathcal{C}$. By induction hypothesis, we know that the set $\{e\in F\backslash \delta(U):e\subseteq U\}$ is contained in the family $\mathcal{F}_{U}$ and the set $\{e\in F\backslash \delta(U):e\subseteq \complement{U}\}$ is contained in the family $\mathcal{F}'_U$. Therefore, the set $F$ is added to the family $\mathcal{F}$ in the second for-loop.
    
    \item Suppose $d(U)=\opt$.
    
    By Theorem \ref{thm:cut-set-recovery}, there exist sets $S\subseteq U$ and $T\subseteq \complement{U}$ with $|S|,|T|\leq 2k-1$ such that the source minimal minimum $(S,T)$-terminal cut $(A,\complement{A})$ satisfies $\delta(A)=\delta(U)=F$. 
    Therefore, the set $A$ is in the collection $\collection$. 
    Since $F=\delta(A)$, the hypergraph $G-\delta(A)$ contains at least $k$ connected components. Therefore, the set $F=\delta(A)$ is added to the family $\mathcal{F}$ in the first for-loop.
\end{enumerate}
Thus, in both cases, we have shown that the set $F$ is contained in the family $\mathcal{F}$. Since the algorithm returns the subfamily of hyperedge sets in $\mathcal{F}$ that are \minkcutsets, the set $F$ is in the family returned by the algorithm.

Next, we bound the run time of the algorithm. Let $N(k,n)$ denote the run time of the algorithm for a $n$-vertex hypergraph. Then, we have $N(1,n)=O(1)$. 
For $k\geq 2$, there are $O(n^{4k-2})$ pairs of subsets $S,T\subseteq V$ with $|S|,|T|\leq 2k-1$ and $S\cap T=\emptyset$. Hence, the first for-loop performs $O(n^{4k-2})$ source minimal minimum $(S,T)$-terminal cut computations. The collection $\mathcal{C}$ and the family $\mathcal{F}$ at the end of the first for-loop each have $O(n^{4k-2})$ sets. This implies that the first for-loop can be implemented to run in $O(n^{4k-2})T(n,p)$ time. For each $A\in\mathcal{C}$, the computation of Enum-Cut-Sets$(G[A],\lfloor k/2\rfloor)$ in the second for-loop runs in $N(\lfloor k/2\rfloor,n)$ time. The computation of Enum-Cut-Sets$(G[\complement{A}],k-\lfloor k/2\rfloor)$ in the second for-loop runs in $N(k-\lfloor k/2\rfloor,n)$ time. Hence, the second for-loop can be implemented to run in $O(n^{4k-2})N(\lfloor k/2\rfloor,n)N(k-\lfloor k/2\rfloor, n)$ time. The last step to prune the family $\mathcal{F}$ can be implemented to run in time that is linear in the time to implement the first and second for-loops: this is because for each member of $\mathcal{F}$, we can decide whether it is a minimum $k$-cut-set in time linear in the time to write this member in $\mathcal{F}$. 
Therefore, we have
\[N(k,n)=O\left(n^{4k-2}\right)T(n,p)+O\left(n^{4k-2}\right)N\left(\left\lfloor \frac{k}{2}\right\rfloor,n\right)N\left(k-\left\lfloor \frac{k}{2}\right\rfloor, n\right).\]
Since $N(1,n)=O(1)$, we have that $N(k,n)=O(n^{(8k-6)\log k})T(n,p)$. 

Finally, we bound the cardinality of the family returned by the algorithm. 
Let $f(k,n)$ be the cardinality of the family returned by the algorithm for a $n$-vertex hypergraph. We note that $f(k,n)$ is at most the cardinality of the family $\family$ computed by the algorithm. 
There are $O(n^{4k-2})$ pairs of subsets $S,T\subseteq V$ with $|S|,|T|\leq 2k-1$ and $S\cap T=\emptyset$. Hence, the total cardinality of the collection $\collection$ and the family $\mathcal{F}$ at the end of the first for-loop is $O(n^{4k-2})$. Consequently, for $k\ge 2$, by the recursion, we have that 
\[
f(k,n)=O(n^{4k-2})f\left(\left\lfloor \frac{k}{2}\right\rfloor,n\right)f\left(k-\left\lfloor \frac{k}{2}\right\rfloor, n\right)
\]
and $f(1,n)=1$. So, $f(k,n)=O(n^{(8k-6)\log k})$.

\end{proof}

We recall that the number of \minkcutsets in a $n$-vertex hypergraph is $O(n^{2k-2})$ \cite{CXY19}. Assuming this bound improves the run-time of Algorithm Enum-Cut-Sets$(G,k)$ in Figure \ref{Algo: DC enum-cuts}.

\begin{lemma}\label{lemma:DC-better-run-time}
Algorithm Enum-Cut-Sets$(G,k)$ in Figure \ref{Algo: DC enum-cuts} can be implemented to run in time $O(n^{16k-26})T(n,p)$, where $n$ is the number of vertices, $p$ is the size of the input hypergraph $G$, and $T(n,p)$ denotes the time complexity for computing the source minimal minimum $(s,t)$-terminal cut in a $n$-vertex hypergraph of size $p$.
\end{lemma}
\begin{proof}
Let $N(k,n)$ denote the run time of the algorithm for a $n$-vertex hypergraph. Then, we have $N(1,n)=O(1)$. For $k\geq 2$, there are $O(n^{4k-2})$ pairs of subsets $S,T\subseteq V$ with $|S|,|T|\leq 2k-1$ and $S\cap T=\emptyset$. Hence, the first for-loop performs $O(n^{4k-2})$ source minimal minimum $(S,T)$-terminal cut computations. 
The collection $\mathcal{C}$ and the family $\mathcal{F}$ at the end of the first for-loop each have $O(n^{4k-2})$ sets. This implies that the first for-loop can be implemented to run in $O(n^{4k-2})T(n,p)$ time. For each $A\in\mathcal{C}$, the computation of Enum-Cut-Sets$(G[A],\lfloor k/2\rfloor)$ in the second for-loop runs in $N(\lfloor k/2\rfloor,n)$ time. The computation of Enum-Cut-Sets$(G[\complement{A}],k-\lfloor k/2\rfloor)$ in the second for-loop runs in $N(k-\lfloor k/2\rfloor,n)$ time. We recall that $\family_A$ consists of all minimum $\lfloor k/2\rfloor$-cut-sets in a $n$-vertex graph and hence, has size $O(n^{2\lfloor k/2\rfloor-2})$. Similarly, $\family'_A$ has size $O(n^{2(k-\lfloor k/2\rfloor)-2})$. 
Hence, the second for-loop can be implemented to run in time
\[
O\left(n^{4k-2}\right)\left(N\left(\left\lfloor \frac{k}{2}\right\rfloor,n\right)+N\left(k-\left\lfloor \frac{k}{2}\right\rfloor, n\right)+O(n^{2\lfloor k/2\rfloor-2})O(n^{2(k-\lfloor k/2\rfloor)-2})\right)
\]
\[
=O\left(n^{4k-2}\right)\left(N\left(\left\lfloor \frac{k}{2}\right\rfloor,n\right)+N\left(k-\left\lfloor \frac{k}{2}\right\rfloor, n\right)+O(n^{2k-4})\right).
\]
Moreover, the size of the family $\mathcal{F}$ at the end of the second for-loop is $O(n^{4k-2})+O(n^{4k-2})|\mathcal{F}_A|\cdot|\mathcal{F}'_A|= O(n^{-2})O(n^{2k-4})=O(n^{6k-6})$. Hence, the last step to prune the family $\mathcal{F}$ can be implemented to run in time $O(n^{6k-6})$. 
Hence, the second for-loop and the last step can together be implemented to run in time 
\[
O\left(n^{4k-2}\right)\left(N\left(\left\lfloor \frac{k}{2}\right\rfloor,n\right)+N\left(k-\left\lfloor \frac{k}{2}\right\rfloor, n\right)+O(n^{2k-4})\right).
\]
Therefore, we have
\[
N(k,n)
=O\left(n^{4k-2}\right)\left(T(n, p) + N\left(\left\lfloor \frac{k}{2}\right\rfloor,n\right)+N\left(k-\left\lfloor \frac{k}{2}\right\rfloor, n\right)+O(n^{2k-4})\right).
\]
Solving the recursive relation gives $N(k,n)=O(n^{16k-26})T(n,p)$. 
\end{proof}

\section{Algorithm for \enumMMh}\label{sec:enummmh-algo}

In this section, we design a 
deterministic algorithm for \enumMMh that runs in time $n^{O(k^2)}p$, where $n$ is the number of vertices 
and $p$ is the size of the input hypergraph. For this, we rely on the notion of $k$-cut-set representatives. 

We recall that for a $k$-partition $(V_1, \ldots, V_k)$ and disjoint subsets $U_1, \ldots, U_k\subseteq V$, the $k$-tuple $(U_1, \ldots, U_k)$ is defined to be a \emph{$k$-cut-set representative} of $(V_1, \ldots, V_k)$ if $U_i\subseteq V_i$ and $\delta(U_i)=\delta(V_i)$ for all $i\in [k]$. 
We first show that there exists a polynomial-time algorithm to verify whether a given $k$-tuple $(U_1, \ldots, U_k)$ is a $k$-cut-set representative. 


\begin{theorem}\label{Thm: hyp minmax-1}
Let $G=(V,E)$ be a $n$-vertex hypergraph of size $p$ and let $k$ be a positive integer. 
Then, there exists an algorithm that takes as input the hypergraph $G$ and disjoint subsets $U_1, \ldots, U_k\subseteq V$ and runs in time $O(knp)$ to decide if 
$(U_1, \ldots, U_k)$ is a $k$-cut-set representative of some $k$-partition $(V_1, \ldots, V_k)$
and if so, then return such a $k$-partition. 
\end{theorem}

\begin{proof}
We will use Algorithm Recover-Partition$(G, U_1, \ldots, U_k)$ in Figure \ref{Algo: hyp minmax-1}.

\begin{figure*}[ht]
\centering\small
\begin{algorithm}
\textul{Algorithm Recover-Partition$(G=(V,E),U_1,\ldots,U_k)$}\+
\\{\bf Input:} Hypergraph $G=(V,E)$ and disjoint subsets $U_1, \ldots, U_k\subseteq V$
\\{\bf Output:} Decide if there exists a $k$-partition $(V_1,\ldots,V_k)$ of $V$ \+\+\\
$\,\;$with $U_i\subseteq V_i$ and $\delta(U_i)=\delta(V_i)$ $\forall i\in[k]$, 
and return one if it exists\-\-
\\Initialize $P_i\leftarrow U_i$ for all $i\in [k]$ 
\\Let $C_1,\ldots,C_t\subseteq V$ be the components of $G-\cup_{i=1}^k\delta(U_i)$ that are disjoint from $\cup_{i=1}^k U_i$
\\For $j=1,\ldots,t$\+
\\If $\exists i\in [k]$ such that $\delta(P_i\cup C_j)=\delta(P_i)$\+
\\$P_i\gets P_i\cup C_j$\-\- 
\\If $(P_1,\ldots,P_k)$ is a $k$-partition of $V$\+
\\Return $(P_1,\ldots,P_k)$\-
\\Else\+
\\Return NO
\end{algorithm}
\caption{Algorithm in Theorem \ref{Thm: hyp minmax-1}}
\label{Algo: hyp minmax-1}
\end{figure*}

We begin by showing correctness. 
Since Algorithm \ref{Algo: hyp minmax-1} maintains $U_i\subseteq P_i$ and $\delta(U_i)=\delta(P_i)$ for all $i\in[k]$, if it returns a $k$-partition, then the $k$-partition necessarily satisfies the required conditions. Next, we show that if $(U_1, \ldots, U_k)$ is a $k$-cut-set representative of a $k$-partition $(V_1, \ldots, V_k)$, then the algorithm will indeed return a $k$-partition $(P_1, \ldots, P_k)$ with $U_i\subseteq P_i$ and $\delta(U_i)=\delta(P_i)$ for all $i\in [k]$ 
(however, $(P_1,\ldots, P_k)$ may not necessarily be the same as $(V_1, \ldots, V_k)$). 

Let $(V_1, \ldots, V_k)$ be a $k$-partition such that $U_i\subseteq V_i$ and $\delta(U_i)=\delta(V_i)$ for all $i\in [k]$. Let $(P_1, \ldots, P_k)$ be the sequence of subsets at the end of the for-loop. Moreover, for each $j\in[t]$, let $(P^j_1, \ldots, P^j_k)$ be the sequence of subsets at the end of the $j$th iteration of the for-loop. For notational convenience, for $i\in [k]$, we will define $P_i^0:=U_i$. 
We note that $(P_1, \ldots, P_k)=(P^t_1, \ldots, P^t_k)$ and that $P_i^0\subseteq P_i^1\subseteq\ldots\subseteq P_i^t$ for every $i\in[k]$. 
We observe that $U_i\subseteq P_i^j$ and $\delta(U_i)=\delta(P_i^j)$ for all $j\in \{0, 1, 2, \ldots, t\}$. Moreover, the subsets $P_1^j, \ldots, P_k^j$ are pairwise disjoint for each $j\in \{0, 1,2, \ldots, t\}$. Therefore, it suffices to show that $\cup_{i=1}^k P_i = V$. 

We claim that $C_1\cup\ldots\cup C_{j}\subseteq\cup_{i=1}^k P_i^{j}$ for each $j\in\{0,1,\ldots,t\}$. Applying this claim for $j=t$ gives that $\cup_{i=1}^k P_i=V$ as desired. We now show the claim by induction on $j$. The base case of $j=0$ holds by definition. 
We now prove the induction step. By induction hypothesis, we have that  $C_1\cup\ldots\cup C_{j-1}\subseteq\cup_{i=1}^k P_i^{j-1}$. 
We will show that there exists $i\in [k]$ such that $\delta(P_i^{j-1}\cup C_j)=\delta(P_i^{j-1})$. 
We know that $C_j$ is contained in one of the sets in $\{V_1, \ldots, V_k\}$, say $C_j\subseteq V_\ell$ for some $\ell\in [k]$. 
We will prove that $\delta(P_{\ell}^{j-1}\cup C_j)=\delta(P_{\ell}^j)$ to complete the proof of the claim. 
Since $C_j$ is a component of $G-\cup_{i=1}^k\delta(U_i)=G-\cup_{i=1}^k\delta(V_i)$, we know that each hyperedge in $\delta(C_j)$ crosses the $k$-partition $(V_1, \ldots, V_k)$. Moreover, each hyperedge in $\delta(C_j)$ intersects $C_j\subseteq V_\ell$, and hence,  $\delta(C_j)\subseteq \delta(V_\ell)$. Therefore,
\[
\delta(P_\ell^{j-1}\cup C_j)-\delta(P_\ell^{j-1})\subseteq \delta(C_j)\subseteq \delta(V_\ell)=\delta(U_\ell)=\delta(P_\ell^{j-1}).
\]
This is possible only if $\delta(P_\ell^{j-1}\cup C_j)-\delta(P_\ell^{j-1})=\emptyset$, i.e., $\delta(P_\ell^{j-1}\cup C_j)\subseteq\delta(P_\ell^{j-1})$. We now show show the reverse inclusion. We have that
\begin{align*}
    \delta(P_\ell^{j-1})-\delta(P_\ell^{j-1}\cup C_j)&=\delta(U_\ell)-\delta(P_\ell^{j-1}\cup C_j)
    \\&=\delta(V_\ell)-\delta(P_\ell^{j-1}\cup C_j)
    \\&=E(V_\ell-P_\ell^{j-1}-C_j,V-V_\ell-P_\ell^{j-1})\cup E(V_\ell\cap P_\ell^{j-1},P_\ell^{j-1}-V_\ell)
    \\&\subseteq E[V-P_\ell^{j-1}]\cup E[P_\ell^{j-1}].
\end{align*}
We note that the LHS is a subset of $\delta(P_\ell^{j-1})$ while the RHS is disjoint from $\delta(P_\ell^{j-1})$ since $E[V-P_\ell^{j-1}]\cap \delta(P_\ell^{j-1})=\emptyset$ and $E[P_\ell^{j-1}]\cap \delta(P_\ell^{j-1})=\emptyset$. Hence, the above containment is possible only if $\delta(P_\ell^{j-1})-\delta(P_\ell^{j-1}\cup C_j)=\emptyset$ and hence, $\delta(P_\ell^{j-1})\subseteq \delta(P_\ell^{j-1}\cup C_j)$. Consequently,  $\delta(P_\ell^{j-1}\cup C_j)=\delta(P_\ell^{j-1})$.




We now bound the run-time. We can verify if there exists $i\in [k]$ such that $\delta(P_i\cup C_j)=\delta(P_i)$ in time $O(kp)$. The number of iterations of the for-loop is $t\le n$. Hence, the total run-time is $O(knp)$. 
\end{proof}

Next, we address the problem of enumerating all \minmaxkcutsets. For this, we define a sub-problem---namely \enumkcutsetreps. The input here is a hypergraph $G=(V,E)$ and a fixed positive integer $k$ (e.g., $k=2, 3, 4, \ldots$). The goal is to 
enumerate a family $\family$ of $k$-cut-set representatives satisfying the following two properties: 
\begin{enumerate}[(1)]
    \item every $k$-tuple $(U_1, \ldots, U_k)$ in the family $\family$ is a $k$-cut-set representative of some optimum $k$-partition $(V_1, \ldots, V_k)$ for \mmh and 
    \item for every optimum $k$-partition $(V_1, \ldots, V_k)$ for \mmh, the family $\family$ contains a $k$-cut-set representative $(U_1, \ldots, U_k)$ of $(V_1, \ldots, V_k)$. 
\end{enumerate}
We note that if a family $\family$ is a solution to \enumkcutsetreps, then returning $\{\cup_{i=1}^k \delta(U_i): (U_1, \ldots, U_k)\in \family \}$ solves \enumMMh. 
Hence, it suffices to solve \enumkcutsetreps in order to solve \enumMMh. 
We describe our algorithm for \enumkcutsetreps in Figure \ref{Algo: hyp minmax-2} and its guarantees in Theorem \ref{Thm: hyp minmax-2}. Theorem \ref{thm:minmax-enumeration} follows from Theorem \ref{Thm: hyp minmax-2}. 

\begin{theorem}\label{Thm: hyp minmax-2}
Let $G=(V,E)$ be a $n$-vertex hypergraph of size $p$ and let $k$ be a positive integer. Then, Algorithm Enum-MinMax-Reps$(G, k)$ in Figure \ref{Algo: hyp minmax-2} 
solves \enumkcutsetreps 
and  
it can be implemented to run in time $O(kn^{4k^2-2k+1}p)$. Moreover, the cardinality of the family returned by the algorithm is $O(n^{4k^2-2k})$.
\end{theorem}

\begin{figure*}[ht]
\centering\small
\begin{algorithm}
\textul{Algorithm Enum-MinMax-Reps$(G=(V,E),k)$}\+
\\{\bf Input:} Hypergraph $G=(V,E)$ and an integer $k\geq 2$
\\{\bf Output:} Family $\mathcal{F}$ of $k$-cut-set representatives of all optimum $k$-partitions\+\+
\\$\,\;$for \mmh\-\- 
\\Initialize $\mathcal{C}\gets\emptyset$, $\mathcal{D}\gets\emptyset$, and  $\mathcal{F}\gets\emptyset$
\\For each pair $(S,T)$ such that $S,T\subseteq V$ with $S\cap T=\emptyset$ and $|S|$, $|T|\leq 2k-1$ \+
\\ Compute the source minimal minimum $(S,T)$-terminal cut $(U,\complement{U})$
\\$\mathcal{C}\gets\mathcal{C}\cup\{U\}$\-
\\For all $(U_1,\ldots, U_k)\in\mathcal{C}^k$ such that $U_1,\ldots, U_k$ are pairwise disjoint\+
\\If Recover-Partition$(G,U_1,\ldots, U_k)$ returns a $k$-partition\+ 
\\$\mathcal{D}\gets\mathcal{D}\cup\{(U_1,\ldots,U_k)\}$\-\-
\\$\lambda \leftarrow \min\{\max_{i\in [k]}d(U_i): (U_1, \ldots, U_k)\in \mathcal{D}\} $
\\For all $(U_1, \ldots, U_k)\in \mathcal{D}$ such that $\max_{i\in [k]}d(U_i)=\lambda$:\+
\\$\mathcal{F}\leftarrow \mathcal{F}\cup\{(U_1, \ldots, U_k)\}$\-
\\ Return $\mathcal{F}$
\end{algorithm}
\caption{Algorithm in Theorem \ref{Thm: hyp minmax-2}}
\label{Algo: hyp minmax-2}
\end{figure*}

\begin{proof}

We begin by showing correctness---i.e., the family $\family$ returned by the algorithm satisfies properties (1) and (2) mentioned in the definition of \enumkcutsetreps. 
By the second for-loop, each $k$-tuple added to the collection $\mathcal{D}$ is a $k$-cut-set representative of some $k$-partition (it need not necessarily be a $k$-cut-set representative of an optimum  $k$-partition for \mmh). The algorithm returns a subfamily of $\mathcal{D}$ and hence, it returns a subfamily of $k$-cut-set representatives. 
We only have to show that a $k$-cut-set representative of an arbitrary optimum $k$-partition for \mmh is present in the family $\mathcal{D}$; this will guarantee that 
the value $\lambda$ computed by the algorithm will exactly be $\optmm$ and owing to the way in which the algorithm constructs the family $\family$ from the family $\mathcal{D}$, it follows that the family $\family$ satisfies properties (1) and (2). 


Let $\optmm$ denote the optimum value of a \minmax $k$-partition in $G$ and let $\optk$ denote the optimum value of a minimum $k$-cut in $G$. We note that $\optmm\le \optk$. This is because, if $(P_1, \ldots, P_k)$ is a $k$-partition with minimum $|\delta(P_1, \ldots, P_k)|$ (i.e., an optimum $k$-partition for \hkcut), then 
\[
\optmm\le \max_{i\in [k]}|\delta(P_i)|\le |\delta(P_1, \ldots, P_k)|=\optk. 
\]

Let $(V_1, \ldots, V_k)$ be an arbitrary optimum $k$-partition for \mmh. We will show that the family $\mathcal{F}$ returned by the algorithm contains a $k$-cut-set representative of $(V_1, \ldots, V_k)$. 
We have that $d(V_i)\leq \optmm\le \optk$ for all $i\in [k]$. Hence, by Theorems \ref{theorem: structure thm 1} and \ref{thm:cut-set-recovery}, there exist subsets $S_i\subseteq V_i$, $T_i\subseteq V-V_i$ with $|S_i|,|T_i|\leq 2k-1$ such that the source minimal minimum $(S_i,T_i)$-terminal cut $(U_i,\complement{U_i})$ satisfies $\delta(U_i)=\delta(V_i)$ for all $i\in [k]$. 
Source minimality of the cut $(U_i,\complement{U_i})$ also guarantees that $U_i\subseteq V_i$ for all $i\in[k]$. Hence, the $k$-tuple $(U_1, \ldots, U_k)$ is a $k$-cut-set representative of $(V_1, \ldots, V_k)$. It remains to show that this $k$-tuple is indeed present in the families $\mathcal{D}$ and $\mathcal{F}$. 
We note that the sets $U_1, \ldots, U_k$ are added to the collection $\mathcal{C}$ in the first for-loop. Since the $k$-tuple $(U_1, \ldots, U_k)$ is a $k$-cut-set representative of the $k$-partition $(V_1, \ldots, V_k)$, the $k$-tuple $(U_1, \ldots, U_k)$ will be added to the family $\mathcal{D}$ in the second for-loop. Since the family $\mathcal{D}$ contains only $k$-cut-set representatives of $k$-partitions, it follows that $\lambda = \optmm$ and $(U_1, \ldots, U_k)$ will be added to the family $\mathcal{F}$ in the third for-loop. Hence, the $k$-cut-set representative $(U_1, \ldots, U_k)$ of the optimum $k$-partition $(V_1, \ldots, V_k)$ for \mmh is present in the family $\mathcal{F}$ returned by the algorithm. 

The bound on the size of the family $\mathcal{F}$ returned by the algorithm is
\[
|\mathcal{F}|\le |\mathcal{D}| \le |\mathcal{C}|^k = O(n^{k(4k-2)}). 
\]

Next, we bound the run time of the algorithm. 
The first for-loop can be implemented to run in time $O(n^{4k-2})T(n, p)$. 
The second for-loop executes the algorithm from Theorem \ref{Thm: hyp minmax-1}  $O(n^{4k^2-2k})$ times and hence, the second for-loop can be implemented to run in time $O(kn^{4k^2-2k+1}p)$. The computation of $\lambda$ and the third for-loop can be implemented to run in time $O(|\mathcal{D}|)=O(n^{4k^2-2k})$. 
Hence, the total run-time is $O(n^{4k-2})T(n,p)+O(kn^{4k^2-2k+1}p)$. We recall that $T(n, p)=O(n p)$ and hence, the total run-time is $O(kn^{4k^2-2k+1}p)$.


\end{proof}


\section{A lower bound on the number of \minmaxkcutsets}
\label{sec:lower-bound}

In this section, we show that there exist $n$-vertex connected graphs for which the number of \minmaxkcutset{s} is $n^{\Omega(k^2)}$. In particular, we show the following result. 

\begin{lemma}
For every positive integer $k\ge 2$, there exists a positive integer $n$ such that the number of optimum $k$-partitions for \mmg in the $n$-vertex complete graph is $n^{\Omega(k^2)}$. 
\end{lemma}
\begin{proof}
Let $k\geq 2$ be fixed, and let $G=(V,E)$ be the complete graph on $n=k(k-1)$ vertices (with all edge weights being uniformly $1$). We will show that $\optmm=(k-1)^3$ and every partition of $V$ into $k$ parts of equal size is an optimum $k$-partition for \mmg. Since the number of partitions of $V$ into $k$ parts of equal size is $\Omega(k^n)=\Omega(n^{k^2/2})$, the lemma follows. 

First, we show that $\optmm\ge (k-1)^3$. For every partition of $V$ into $k$ non-empty parts, the largest part has at least $k-1$ vertices by pigeonhole principle, and at most $k(k-1)-(k-1)=(k-1)^2$ vertices since each of the remaining $k-1$ parts contain at least one vertex. Therefore, the cut value of the largest part is at least
\[\min_{x\in\{k-1,\ \ldots,\ (k-1)^2\}}x(n-x)=(k-1)^3.\]
The equality follows since $n=k(k-1)$ and the function $f(x)=x(n-x)$ is convex and is minimized at the boundaries. This implies that $\optmm\ge (k-1)^3$.

Next, we show that $\optmm \le (k-1)^3$ and every partition of $V$ into $k$ parts of equal size is an optimum $k$-partition for \mmg. 
Let $(V_1,\ldots,V_k)$ be an arbitrary $k$-partition of $V$ such that $|V_i|=k-1$ for all $i\in[k]$. The \mmg objective value of this $k$-partition is $(k-1)(n-(k-1))=(k-1)^3$. Thus, $(V_1,\ldots,V_k)$ is an optimum $k$-partition for \mmg in $G$. 

\end{proof}

We note that our example exhibiting $n^{\Omega(k^2)}$ optimum $k$-partitions for \mmg has the number of vertices $n$ upper bounded by a function of $k$. 
We are not aware of examples that exhibit $n^{\Omega(k^2)}$ optimum $k$-partitions for \mmg for fixed $k$ but arbitrary $n$ (e.g., $k=2,3,4,...$ but $n$ is arbitrary). 

\section{Conclusion}
\label{sec:conclusion}
We showed the first polynomial bound on the number of \minmaxkcutsets in hypergraphs for every fixed $k$ and gave a polynomial-time algorithm to enumerate all \minmaxkcutsets as well as all \minkcutsets in hypergraphs for every fixed $k$. 
Our main contribution is a structural theorem that 
is the backbone of the correctness analysis of our enumeration algorithms. 
In order to enumerate \minmaxkcutsets in hypergraphs, we introduced the notion of $k$-cut-set representatives and enumerated $k$-cut-set representatives of all optimum $k$-partitions for \mmh. 
Our technique builds on known structural results for \hkcut and \mmh \cite{CC20, CC21, BCW22}.

The technique underlying our enumeration algorithms is not necessarily novel---we simply rely on minimum $(s,t)$-terminal cuts. 
Using fixed-terminal cuts to address global partitioning problems is not a novel technique by itself---it is common knowledge that minimum $(s, t)$-terminal cuts can be used to solve global minimum cut. 
However, there are several problems where naive use of this technique fails to lead to efficient algorithms: e.g., multiway cut does not help in solving \gkcut since multiway cut is NP-hard. Adapting this technique 
for specific partitioning problems requires careful identification of structural properties. In fact, beautiful structural properties have been shown for a rich variety of partitioning problems in combinatorial optimization in order to exploit this technique: 
for example, it was used (1) to design the first efficient algorithm for \gkcut \cite{GH94}, (2) to solve certain constrained submodular minimization problems \cite{GR95, NSZ19}, and (3) more recently, to design fast algorithms for global minimum cut in graphs and for Gomory-Hu tree in unweighted graphs \cite{LP20, AKT21}. Our use of this technique also relies on identifying and proving a suitable structural property, namely Theorem \ref{thm:cut-set-recovery}. 
The advantage of our structural property is that it simultaneously enables 
enumeration of \minkcutsets as well as \minmaxkcutsets in hypergraphs which was not possible via  structural theorems that were developed before. 
Furthermore, it helps in showing the first polynomial bound on the number of \minmaxkcutsets in hypergraphs for every fixed $k$. 


We also emphasize a limitation of our technique. Although it helps in solving \enumhkcut and \enumMMh, it does not help in solving a seemingly related hypergraph $k$-partitioning problem---namely, given a hypergraph $G=(V, E)$ and a fixed integer $k$, find a $k$-partition $(V_1, \ldots, V_k)$ of the vertex set that minimizes $\sum_{i=1}^k |\delta(V_i)|$. Natural variants of our structural theorem fail to hold for this objective. Resolving the complexity of this variant of the hypergraph $k$-partitioning problem for $k\ge 5$ remains open. 

We mention an open question concerning \hkcut and the enumeration of \minkcutsets in hypergraphs for fixed $k$. We recall the status in graphs: the number of minimum $k$-partitions in a connected graph was known to be $O(n^{2k-2})$ via Karger-Stein's algorithm \cite{KS96} and $\Omega(n^k)$ via the cycle example, where $n$ is the number of vertices; recent works have improved on the upper bound to match the lower bound for fixed $k$---this improvement in upper bound also led to the best possible $O(n^k)$-time algorithm for \gkcut for fixed $k$ \cite{GLL19-STOC, GLL20-STOC, GHLL20}. For hypergraphs, the number of \minkcutsets is known to be $O(n^{2k-2})$ and $\Omega(n^k)$. Can we improve the upper/lower bound? Is it possible to design an algorithm for \hkcut that runs in time $O(n^kp)$?


\bibliographystyle{amsplain}
\bibliography{references}


\appendix

\end{document}